\documentclass[10pt]{article}

\usepackage[T1]{fontenc}
\usepackage[a4paper, margin=.95in]{geometry}

\usepackage{graphicx}
\usepackage[english]{babel}
\usepackage{amsthm}
\newtheorem{theorem}{Theorem}[section]
\newtheorem{corollary}[theorem]{Corollary}
\newtheorem{lemma}[theorem]{Lemma}
\newtheorem{proposition}[theorem]{Proposition}
\newtheorem{definition}[theorem]{Definition}

\usepackage{thmtools}

\usepackage{amsmath, amsfonts}
\usepackage{enumitem}
\usepackage{bbm}
\usepackage{tikz}
\usepackage[colorlinks, pagebackref]{hyperref}
\def\tmp#1#2#3{%
  \definecolor{Hy#1color}{#2}{#3}%
  \hypersetup{#1color=Hy#1color}}
\tmp{link}{HTML}{800006}
\tmp{cite}{HTML}{2E7E2A}
\tmp{file}{HTML}{131877}
\tmp{url} {HTML}{8A0087}
\tmp{menu}{HTML}{727500}
\tmp{run} {HTML}{137776}
\def\tmp#1#2{%
  \colorlet{Hy#1bordercolor}{Hy#1color#2}%
  \hypersetup{#1bordercolor=Hy#1bordercolor}}
\tmp{link}{!60!white}
\tmp{cite}{!60!white}
\tmp{file}{!60!white}
\tmp{url} {!60!white}
\tmp{menu}{!60!white}
\tmp{run} {!60!white}
\usepackage{subcaption}
\usepackage{bm}
\usepackage{tikz}
\usepackage{pgfplots}
\pgfplotsset{compat=1.18}

\renewcommand{\vec}[1]{\boldsymbol{\mathbf{#1}}}
\renewcommand{\o}{\ensuremath{\Phi}}
\renewcommand{\t}{\ensuremath{\sigma}}
\DeclareMathOperator{\OPT}{OPT}
\DeclareMathOperator{\comp}{comp}
\DeclareMathOperator{\ALG}{ALG}
\DeclareMathOperator{\GRD}{GRD}
\DeclareMathOperator{\SW}{SW}
\DeclareMathOperator{\LW}{LW}
\DeclareMathOperator*{\argmax}{\mathop{argmax}}
\DeclareMathOperator{\supp}{supp}
\DeclareMathOperator*{\E}{\mathop{\mathbb{E}}}
\newcommand{\vtype}{\ensuremath{\vec{\tau}}}
\DeclareMathOperator{\rev}{REV}

\usepackage{natbib}

\usepackage[noblocks]{authblk}

\begin{document}

\title{\bfseries
Generalized Prophet Inequalities for Online Auctions}

\author[1]{Twan Kroll}
\author[1]{Rebecca Reiffenhäuser}

\affil[1]{Institute for Logic, Language and Computation, University of Amsterdam, The Netherlands.}

\date{}

\maketitle
\begin{abstract}
\noindent
Prophet Inequalities for online auctions, as many classical auction settings, traditionally consider agents who strive to maximize quasi-linear utility, i.e. their obtained value minus the price paid.
This simplified view of the agents' objectives does not capture real-world settings where agents are almost always constrained by budgets, and can have different objectives altogether (e.g., very prominently, value maximizers in autobidding).

We generalize classic prophet inequality auction settings to accommodate unknown, varying types of agent objectives and budgets, modeling them as stochastic input similar to the traditional valuation distributions.

In this new model, we analyze (anonymous) posted-price mechanisms, one main category of which are the \emph{balanced prices} central to many auction results in the prophet inequality model.
We demonstrate both the resilience and limits of balanced price mechanisms for generalized prophet inequalities with stochastic agent objectives, budgets and valuations: our generalization of the balanced prices framework incorporates budget constraints and mixed objectives. We show constant competitive ratios for some the most central settings, spanning all the way up to agents with subadditive valuation functions. On the way, we establish trade-offs between competitive ratio and agents' price-sensitivity, and give a general reduction from the budgeted setting to budget-free value maximizers.
\end{abstract}
\tableofcontents

\thispagestyle{empty}
\newpage
\addtocounter{page}{-1}
\section{Introduction} 
\subsection{Background and Motivation}
Prophet inequalities (PIs) span a vast field of problem variants and models, with some of their most central applications in \emph{online auctions} \cite{feldman2014,dutting2024,correa2023constant,braun2023simplified,alaei2014}. Here, a set of items is available for sale, and strategic buyers arrive in adversarial order. An online (posted-price) mechanism sets a (either fixed, or per-round) price for each possible bundle of items, with the aim of maximizing the sum of agents' values, the so-called \emph{social welfare}. 
Each buyer will, on arrival, draw their valuation from a known independent distribution and be allowed to purchase the bundle of available items that maximizes their personal objective, given the posted prices. The vast majority of literature assumes that buyers maximize \emph{quasilinear utility}, i.e., value of their assigned outcome minus the price paid.

However, the assumption that buyers are utility maximizers is often unrealistic.
Most prominently, the rise of \emph{autobidding} in ad auctions has made clear that beyond spending less than they gain, real agents might not care much about the price paid and instead strive to generate the highest-possible value \cite{aggarwal2019autobidding,aggarwal2024,balseiro2021landscape}. 
Complementing the existence of an optimal, 2-competitive mechanism for XOS combinatorial auctions with utility-maximizing agents \cite{feldman2014}, for value maximizers with XOS valuations, there exist prices that achieve an expected social welfare at least a quarter of the expected optimal social welfare \cite{deng2022}. 

Both results are obtained via \emph{posted-price} mechanisms \cite{chawla2010multi} as described above. Such mechanisms have many desirable properties: they are intuitive, algorithmically simple, and inherently \emph{truthful}, meaning that agents have no incentive to misreport their valuation. 

However, neither established model reaches far enough. Consider, e.g., an auction for spots on the recommendation lists of a number of websites.
While there might be private buyers interested in being featured on the very few, most profitable specific sites featured in the current auction, they might compete with large companies that routinely take part in a high number of such auctions via automated bidding systems. 
While agents of both kinds might be interested in simply maximizing their notoriety using a given advertising budget, they might also have outside options for advertising with certain value/price paid, leaving them to some degree price-sensitive. 
Finally, their valuations can be complex and, even when auto-bidding, far from additive: companies will, for example, aim to be recommended by a wide range of websites instead of several ones on the same topic. Reasonably, future autobidding systems will advance to accommodate much more detailed valuations.
Both phenomena, the presence of classic together with auto-bidding agents and intermediate agents, as well as value-maximizers with complex valuations, have so far rarely been considered. In this work, we introduce the first prophet inequalities accommodating all kinds of such different agents, at the same time.

No matter the agents' objectives and class of complement-free valuations considered, in practically all settings, it is also reasonable to assume that buyers are constrained by a budget. 
In the presence of budgets, the standard mechanism objective changes from sum of valuations to \emph{liquid welfare}, introduced in \cite{dobzinski2014efficiency}. Liquid welfare represents agents' \emph{ability to pay}, and is defined by capping the value of an assignment by the receiving agent's budget. For budget-constrained utility maximizers with XOS valuations, there exists a posted price mechanism that achieves an expected liquid welfare of at least a quarter of the expected offline optimal liquid welfare \cite{fotakis2019}.
Beyond the standard utility maximizers, budgets 
under the liquid-welfare benchmark have to the best of our knowledge not been studied in prophet inequalities. 
Accommodating both natural extensions of different agent objectives and budgets, we study a prophet inequality auction model where besides the valuation, each agent's objective and budget are also stochastic.

\subsection{Our results}
We introduce prophet inequalities for agents with stochastic objectives, i.e. which strive to each maximize a function of the valuation and price for their assigned outcome, drawn from a known distribution at the time of arrival. 
We study posted-price mechanisms in which agents face price listings for obtainable outcomes, like being assigned certain (sets of) items, and then choose an outcome maximizing their objective. Following the majority of the literature (see related work) and major applications, in our model, agents have varying objectives according to their \emph{price sensitivity}, parametrized by a $\sigma\in [0,1]$. This captures both the prevalent notions of \emph{value-maximizers} and \emph{utility-maximizers} as the extreme cases.
Furthermore, while seeking to maximize their objective, agents are constrained by budgets and a Return-on-Spend constraint as is common in autobidding. In the budget-free case, the mechanism strives to maximize the standard objective of social welfare, while for settings with budgets, we instead analyze the standard notion of liquid welfare (see e.g. \cite{dobzinski2014efficiency,aggarwal2024}). For budget-free settings, liquid welfare coincides with social welfare.

\paragraph{Budget Reduction.}
To obtain welfare guarantees in budgeted settings, we prove that budget constraints can be removed via performing an appropriate adjustment of agents' valuations and objective.
Our reduction idea is different from, but not unsimilar to the one in \cite{baldeschi2026}. 
Consider a distribution $\mathcal D$ over $(\vec v,\vec B,\vec \t)$, i.e., $\mathcal D$ is a distribution over valuation, budget and price sensitivity profiles where we assume independence across agents. We define a budget-free distribution $\hat{\mathcal D}$ as follows. Sample $(\vec v,\vec B,\vec \t)\sim \mathcal D$, then the value in the reduced realization becomes the budget-capped valuation of the original instance, i.e., one caps the valuation of the original instance by the budget. However, this transformation does not preserve the agent's decisions, even if the prices are preserved. 
In realizations where an agent's valuation exceeds their budget, capping the valuation would alter the objective of the agent. 
This means that the outcome that maximizes the original objective might not maximize the agent's new objective with capped valuations. 
To remedy this, we transform the agent's objective in such realizations to become value maximizing. Note that this naturally creates stochastic objectives: even if in the original instance, all agents were e.g. utility maximizers, the transformed instance includes agents that are value maximizers in some realizations, creating settings with mixed objectives. 

\paragraph{Utility-Bounded Prices} Inspired by offline smoothness \cite{syrgkanis2013composable}, we introduce the notion of \emph{$(\lambda,\mu)$-utility-bounded} (UB) pricing rules. Many guarantees of posted-price mechanisms exhibit a proof structure based on a balance or trade-off between revenue and utility. 
Similar to this, our analysis considers a \emph{tentative} feasible bundle each agent could have deviated to, and uses these deviations to lower-bound the aggregate proxy utility of the resulting allocation 
by some difference between the optimal welfare and revenue scaled by $\lambda$ and $\mu$, respectively. The notion of utility-boundedness expresses and quantifies this property and can thus capture many existing pricing rules in the literature. Let $\sigma_{\min}$ be the smallest price sensitivity in any realization of a budget-free distribution $\mathcal D$ and 0 if budgeted. 

\begin{theorem}[Informal version of Theorem~\ref{thm:utility-bounded-guarantee}]\label{thm:utility-bound-guarantee-informal}
    Any pricing rule that is $(\lambda,\mu)$-UB achieves a liquid welfare of at least 
    $$\frac{\lambda}{\max\{1,1+\mu-\sigma_{\min}\}}\E_{(\vec v,\vec B,\vec \t)\sim \mathcal D}[\LW(\OPT(\vec v,\vec B))],$$
    where $\OPT(\vec v,\vec B)$ is the offline liquid welfare maximizing allocation under the valuation and budget profiles $\vec v$ and $\vec B$, respectively.
\end{theorem}
This theorem has several interesting applications. First, from this result one can easily prove that the competitive ratio of $(\lambda,\mu)$-UB pricing rules in budget-free utility maximizing settings deteriorates with at most a factor of $2-\sigma_{\min}$ when applying this pricing rule to our general model. More specifically, budgets and value-maximizers worsen our guarantees by at most a factor of 2.

\paragraph{Combinatorial Auctions} One of the central applications/settings captured by prophet inequalities is that of combinatorial auctions. In a combinatorial auction, there is a set $M$ of items. For each agent, the outcome set is $2^M$ and allocations are feasible if all the bundles are pairwise disjoint. Let $\mathcal I_{\Sigma}$ denote the class of all instances in which agents' price sensitivities are sampled from some distribution over $\Sigma\subseteq [0,1]$ and $\mathcal I^{\infty}_{\Sigma}$ denote its budget-free subclass. 

\paragraph{XOS Combinatorial Auctions}
When agents have XOS valuations, using the additive support of a valuation profile $\vec v$ for the optimal offline allocation, the price of item $j$ can be 
set to half its expected contribution to an offline optimum, defined via the additive support value of $j$'s winning agent. This classic pricing rule results in an optimal competitive ratio of $2$, see \cite{feldman2014}. Since $\min(v(\cdot),B_i)$ remains XOS when $v(\cdot)$ is XOS, this naturally extends to budgeted settings and the liquid welfare benchmark. Moreover, instead of the factor of a half, we can consider any factor $\delta\in (0,1)$ and let $\vec p_{\delta}$ denote this pricing rule.

For general $\delta\in(0,1)$, we prove that $\vec p_{\delta}$ is $(1-\delta,\frac{1-\delta}{\delta})$-UB. Theorem~\ref{thm:utility-bound-guarantee-informal} then allows us to upper bound the competitive ratio of the pricing rule $\vec p_{\delta}$ over the class $\mathcal I_{\Sigma}^{\infty}$, for any $\delta\in(0,1)$, as
\begin{equation}\label{eq:CA-intro-tight}\comp(\mathcal I_{\Sigma}^{\infty},\vec p_{\delta})\leq \max\left\{\frac{1-\delta\cdot \sigma_{\min}}{\delta(1-\delta)},\frac{1}{1-\delta}\right\}.\end{equation}
We prove that this is tight for every $\delta\in (0,1)$ by constructing matching lower bounds. The pricing rule $\vec p_{\frac12}$, often called the standard \emph{half-pricing rule}, has been widely researched before. E.g. for budget-free combinatorial auctions with XOS-valuations, it was considered in \cite{Dutting2020} and \cite{deng2022}. Using $\delta=\frac12$ in the bound in Equation~\eqref{eq:CA-intro-tight} gives a bound that deteriorates linearly in $\sigma_{\min}$
$$\comp(\mathcal I_{\Sigma}^{\infty},\vec p_{\frac12})=4-2\sigma_{\min}.$$
When $\Sigma=\{1\}$ and $\delta=\frac12$, we recover the problem studied in \cite{feldman2014} where they prove a competitive ratio of 2, which our result matches. 
The setting in \cite{deng2022} matches ours when $\Sigma=\{0\}$ and $\delta=\frac12$, for which they proved an upper bound of 4 (again matching our results). 
Previously, to the best of our knowledge, this latter bound was not known to be tight.

The optimal $\delta$ in terms of $\t_{\min}$ can be computed from the tight bound in Equation~\eqref{eq:CA-intro-tight}. The way in which the optimal $\delta$, call it $\delta^\ast$, varies with respect to $\t_{\min}$ is captured in Figure~\ref{fig:optimal-delta}. 
The deterioration of the competitive ratio of the pricing rule $\vec p_{\delta^\ast(\t_{\min})}$ in terms of $\t_{\min}$ is illustrated in Figure~\ref{fig:comparison} (violet curve) in comparison to $4-2\t_{\min}$ (red curve).

\begin{figure}
\newlength{\mywidth}
\centering
\begin{subfigure}[t]{0.48\textwidth}
\caption{Competitive ratios for all $\t_{\min}$, for optimal choices $\delta^*$ vs. $\delta=1/2$ (i.e. the established half-pricing rule).}
\setlength{\mywidth}{0.84\textwidth}
  \begin{tikzpicture}
  [ declare function={
    func1(\x)= (4-2*\x);
    func2(\x)=and(\x>= 0.61803398875,\x<=1) * (1/(1+\x))   +
    and(\x>0, \x<0.61803398875) * (1-sqrt(1-\x))/(\x)+
    (\x==0) * (1/2);
    func3(\x)=and(\x>= 0.61803398875,\x<=1) * (1/(1-func2(\x)))   +
    and(\x>0, \x<0.61803398875) * (1-\x * func2(\x))/(func2(\x) * (1-func2(\x)))+
    (\x==0) * (4);}
    ]
    \begin{axis}[
      scale only axis,
      width=\mywidth,
      height=.75\mywidth,
      xmin=0, xmax=1.07,
      ymin=1.9, ymax=4.2,
      axis x line=bottom,
      axis y line=left,
      axis line style={line width=1pt},
      tick style={semithick},
      ticklabel style={font=\small},
      xlabel={$\t_{\min}$},
      xlabel style={at={(axis description cs:0.5,-0.12)},anchor=north},
      ylabel={},
      ylabel style={xshift=2.2cm, yshift=-1.2cm, rotate=-90},
      xtick={0,1},
      xticklabels={$0$, $1$},
      ytick={2,4},
      yticklabels={$2$,$4$},
      clip=false,
    ]
     \addplot[red, domain=0:1, samples = 100, very thick]{func1(x)}; 
     \addplot[blue, domain=0:1, samples = 100, very thick]{func3(x)}; 
      \addplot[only marks, mark=*] coordinates {(0,4)};
      \addplot[only marks, mark=*, node near coord style={anchor= south west}, nodes near coords={(1,2)}] coordinates {(1,2)};
      
    \node at (0.9,2.6) {$4-2\t_{\min}$}; 
    \node at (0.55,2.1) {$\comp(\mathcal I_{\Sigma}^{\infty},\vec p_{\delta^\ast})$};
    \end{axis}
  \end{tikzpicture}
  \label{fig:comparison}
\end{subfigure}
\hfill
\begin{subfigure}[t]{0.45\textwidth}
\caption{Optimal factors $\delta^*$ for all $\t_{min}$, replacing the established half-pricing rule from standard utility- or value-maximizers. }
\setlength{\mywidth}{0.84\textwidth}
  \begin{tikzpicture}
              [ declare function={
    func2(\x)= and(\x>= 0.61803398875,\x<=1) * (1/(1+\x))   +
    and(\x>0, \x<0.61803398875) * (1-sqrt(1-\x))/(\x)+
    (\x==0) * (1/2);}
    ]

    \begin{axis}[
      scale only axis,
      width=\mywidth,
      height=.75\mywidth,
      xmin=0, xmax=1.07,
      ymin=0.35, ymax=0.75,
      axis x line=bottom,
      axis y line=left,
      axis line style={line width=1pt},
      tick style={semithick},
      ticklabel style={font=\small},
      xlabel={$\t_{\min}$},
      xlabel style={at={(axis description cs:0.5,-0.12)},anchor=north},
      ylabel={},
      ylabel style={xshift=2.2cm, yshift=-1.2cm, rotate=-90},
      xtick={0,0.61803398875,1},
      xticklabels={$0$,$\frac{\sqrt{5}-1}{2}$, $1$},
      ytick={0.5,0.61803398875},
      yticklabels={$\frac12$,$\frac{\sqrt{5}-1}{2}$},
      clip=false,
    ]
     \addplot[red, domain=0:1, samples = 100, very thick]{func2(x)}; 
    \node at (0.88,0.6) {$\delta^\ast(\t_{\min})$};  
    \end{axis}
  \end{tikzpicture}
  \label{fig:optimal-delta}
\end{subfigure}
\end{figure}

For the class with budget-constrained settings, $\mathcal I_{\Sigma}$, Theorem~\ref{thm:utility-bound-guarantee-informal} implies the following bound on the competitive ratio of $\vec p_{\delta}$ for $\delta\in (0,1)$:
\begin{equation}\label{eq:budget-bound-CA}\comp(\mathcal I_{\Sigma},\vec p_{\delta})\leq \frac{1}{(1-\delta)\delta}.\end{equation}
For any $\delta\in (0,1)$ we construct a matching lower bound to prove that the bound in Equation~\eqref{eq:budget-bound-CA} is tight. The competitive ratio is minimized for $\delta=\frac12$, yielding $4$.
When $\Sigma=\{1\}$, our setting matches that studied in \cite{fotakis2019}, where they obtain an upper bound of $4$ specific to the pricing rule $\vec p_{\frac12}$. This upper bound for budgeted utility maximizers with budgets was, to the best of our knowledge, not known to be tight for the respective pricing rule.

\paragraph{Subadditive Combinatorial Auctions} For subadditive combinatorial auctions, the existence of constant-competitive prophet inequalities was only recently established. Based on the results of \cite{correa2023constant}, we show the existence of dynamic bundle prices that are $(\frac16,1)$-UB showing by Theorem~\ref{thm:utility-bound-guarantee-informal} that these prices are $6(2-\sigma_{\min})$-competitive for subadditive combinatorial auctions with stochastic agent objectives, further demonstrating the power and generality of our approach. Since $\sigma_{\min}=0$ for budgeted settings, it immediately follows that this pricing rule is $12$-competitive when agents are budget-constrained.

The correspondence in upper bounds between budget-free value maximizers and the budgeted setting is no coincidence, and our reduction technique sheds light on this phenomenon. 

\paragraph{Balanced Prices}
For \emph{utility maximizers}, and far beyond just auction settings, the concept of \emph{balanced prices} introduced 
by \cite{Dutting2020} is a central tool for obtaining posted-price prophet inequalities and fuel a range of important results in prophet inequalities. Balanced prices satisfy two key inequalities, representing that the total price paid for an overall outcome $x$ corresponds to at least an $1/\alpha$-fraction of the loss in possible welfare after choosing $x$ (prices are not too low), but at most a $\beta_1$-fraction of the optimal welfare achievable after choosing $x$, plus $\beta_2$ times the overall algorithm welfare (prices are not too high). If the pricing rule $\vec p^{(\vec v,\vec \t)}$ satisfies both, it is called $(\alpha,\beta_1,\beta_2)$-balanced. We extend balanced prices to the budgeted setting and prove that they are also captured by the utility-boundedness notion. More specifically, if a pricing rule is $(\alpha,\beta_1,\beta_2)$-balanced and satisfies some additional mild assumptions, then the expected pricing rule scaled by $\delta$ is $(1-\delta\beta_1-\delta\beta_2,\alpha \frac{1-\delta\beta_1}{\delta})$-UB. This in combination with Theorem~\ref{thm:utility-bound-guarantee-informal} allows us to import many results.

\paragraph{MPH-k and Matroid Constrained Auctions with Budgets}
Mirroring the results for utility maximizers in \cite{Dutting2020}, we extend balanced prices to both budgeted combinatorial auctions with MPH-$k$ valuations, introduced in \cite{FeigeFIILS15}, and to matroid-constraint auctions, also exploring computational aspects. 
For the budget-free setting with utility maximizers, \cite{Dutting2020} give an upper bound of $4k-2$ for combinatorial auctions with MPH-$k$ valuations.
Similarly, for budget-free auctions with matroid constraints, they show existence of a $2$-approx. pricing rule w.r.t. $\OPT$. Finally, they construct prices that can be computed in polynomial time, and yield a $4$-approximation. Applying our framework to these settings, for budgeted combinatorial auctions with MPH-$k$ valuations, we achieve an upper bound of $4k$.
For matroid-constrained auctions, based on their existential proof we get a $4$-approximation, while their computationally efficient pricing rule can be extended to budgeted settings to still achieve an $8$-approximation for submodular valuations.

\subsection{Further Related Work}\label{sec:further-related-works}
\paragraph{Prophet Inequalities} Initiated by \cite{krengel1977,krengel1978}, prophet inequalities are a central concept in optimal stopping theory. They proved existence of a strategy that in expectation achieves at least half of the expected reward of a prophet who knows all realizations in advance. \cite{Cahn1984} showed that the same bound of 1/2 can be obtained by accepting any value that is at least a suitable threshold. 
In the past decades, prophet inequalities have garnered increased interest from the TCS community, in large part due to their connection with online mechanism design, first observed by \cite{hajiaghayi2007automated}. \cite{chawla2010multi} developed a more general theory connecting prophet inequalities to sequential posted prices.

In fact, \cite{correa2019pricing} proved an equivalence between posted price mechanisms and prophet inequalities based on threshold strategies. Later, \cite{banihashem2024power} gave a general reduction to show that any prophet inequality has a posted price implementation. 

Closest to us in topic are online combinatorial auctions, where buyers arrive sequentially and choose bundles of items given posted prices. \cite{feldman2014} prove that for XOS valuations, a generalization of the original half-pricing rule is 2-competitive. Retaining constant competitive ratios, \cite{dutting2024} show that few samples of the distributions suffice for XOS combinatorial auctions. For subadditive buyers, combining \cite{feldman2014} with \cite{bhawalkar2011welfare} yields an $O(\log m)$ approximation (with $m$ being the number of items) for subadditive buyers. Later, this ratio was improved in \cite{zhang2022improved} to $O(\log m / \log \log m)$ which subsequently improved to $O(\log \log m)$ in \cite{dutting2020log} when finally a constant factor of 6 was achieved in \cite{correa2023constant}.

\paragraph{Balanced prices.} The balanced prices framework, introduced in \cite{Dutting2020}, offers a unifying perspective for many existing prophet inequalities. Its versatility allowed for deriving new and improved results in settings like online knapsack, matroids and integer packing problems. Using balancedness to construct prices/thresholds was an idea already explored by \cite{kleinberg2012matroid} through the notion of balanced thresholds. Later works, such as \cite{dutting2019posted}, extended balanced prices to give constant guarantees in limited information settings where the exact distribution of the buyer's valuation is unknown. 

\paragraph{Online Combinatorial Auctions and Autobidding.} \cite{aggarwal2019autobidding} initiated the study of autobidding settings and the model has since been analyzed through the lens of efficiency by \cite{deng2021towards,balseiro2021,mehta2022auction,deng2024}. For posted price mechanisms, \cite{deng2022} studied combinatorial auctions when all agents are RoS-constrained value maximizers and proved that the half-pricing rule is $4$-competitive up to XOS valuations. \cite{fotakis2019} studies posted price mechanisms for XOS combinatorial auctions with budget-constrained utility maximizers and proved that the half-pricing rule is 4-competitive.

\paragraph{Agents' objectives.} We go beyond settings in which all agents have the same objective, and allow the objectives to even be stochastic. The price sensitivity parametrization, $\sigma\in [0,1]$, of an agent's objective was considered by \cite{balseiro2021} and \cite{jiang2026revenue}. Relatedly, \cite{baldeschi2026} derive upper bounds on the price of anarchy of simultaneous first-price auctions in this generalized agent model by reducing budget-constrained instances to budget-free instances. Their bounds depend on both $\sigma_{\min}:=\min\Sigma$ and $\sigma_{\max}:=\max\Sigma$, while our bounds for $\sigma \in [0,1]$ depend only on $\sigma_{\min}$.

\section{Preliminaries} 
\paragraph{Problem Formulation} Our prophet inequality setting is similar to that of \emph{combinatorial allocation problems}
presented in \cite{Dutting2020}, but generalizes agents' objectives. In this setting, there is a set $N:=[n]$ of agents and for each agent $i\in [n]$ there is an outcome set $X_i$ containing $\emptyset$. 
The \emph{joint outcome space} and \emph{outcome profile} are denoted by $X=X_1\times\dots\times X_n$ and $\vec x=(x_i)_{i\in [n]}$, where $x_i\in X_i$, respectively.
Given an outcome profile $\vec x$ and a set of agents $S$, we use $\vec x_S$ to denote the \emph{partial outcome profile} in which each agent $i\in S$ receives $x_i$ and each $i\notin S$ receives $\emptyset$. Specifically, given an agent $i\in [n]$, we write $\vec x_{[i-1]}$ to denote the partial outcome profile where $i'<i$ receives $x_{i'}$ and each agent $i'\geq i$ receives $\emptyset$. 
There is a subset $\mathcal F\subseteq X$ of \emph{feasible outcomes} for which we assume that it is downward closed, i.e., if $\vec x\in \mathcal F$, then $\vec x_S\in \mathcal F$ for any $S\subseteq [n]$. 

\paragraph{Posted Price Mechanisms} A \emph{pricing rule} is a profile of functions $\vec p=(p_1,\dots p_n)$ that assigns nonnegative prices to outcomes. Let $p_i(x_i\mid\vec y)$ denote the price offered to agent $i$ for receiving outcome $x_i$, given the partial outcome profile $\vec y$. We require that $p_i(x_i\mid\vec y)=\infty$ whenever $(x_i,\vec y_{-i})\not\in \mathcal F$ and $p_i(\emptyset\mid \vec y)=0$. If prices do not depend on the partial outcomes (apart from feasibility) they are referred to as \emph{static}, and if they do not depend on the agent's identity they are called \emph{anonymous}. We assume w.l.o.g. that agents arrive in the order of their index, i.e. agent $i$ arrives before agent $j$ whenever $i<j$. Upon arrival and according to the prices $p_i$, agent $i$ chooses the outcome $x_i\in X_i$ that maximizes their objective $\o_i$.

\paragraph{Agents' Objectives} Each agent $i$ has a monotone valuation function $v_i:X_i\to \mathbb R_{\geq0}$, a budget $B_i\in (0,\infty]$, and a price sensitivity $\t_i\in [0,1]$. 
The valuations of agents are contained in some class $\mathcal V$ and we assume for the remainder of this paper that this class is closed under capping. More precisely, if $v_i\in \mathcal V$, then for any constant $B_i\in [0,\infty]$ we have that $\hat v_i(\cdot):=\min\{v_i(\cdot),B_i\}$ lies in $\mathcal V$.
We further assume that each agent type is specified by its valuation, budget, and price sensitivity. The type of agent $i$ is drawn from a distribution $\mathcal D_i$. Let $\mathcal D= \mathcal D_1\times\dots \times \mathcal D_n$ denote the product distribution over type profiles. Abusing notation slightly, we denote a sample of $\mathcal D$ as $\vtype =(\vec v,\vec B,\vec \t)\sim \mathcal D$ and write $\pi$ for the marginal distribution of $\mathcal D$ over $(\vec v,\vec B)$.
Also, we assume that distribution $\mathcal D$ is public knowledge, and pricing rules may depend on it. We often suppress the dependence on $\mathcal D$ to simplify notation. 
Given a partial outcome $\vec x_{[i-1]}$, agent $i$ chooses the outcome that maximizes their objective. This is captured by the following program
$$\begin{aligned}
    \max_{x_i\in X_i}\qquad & v_i(x_i)-\t_i\cdot p_i(x_i\mid\vec x_{[i-1]})  \\
    \text{subject to}\qquad  &  p_i(x_i\mid\vec x_{[i-1]})\leq v_i(x_i)\\
    &p_i(x_i\mid\vec x_{[i-1]})\leq B_i,
\end{aligned}$$
where we let $\o_i(\vec x):=v_i(x_i)-\t_i\cdot p_i(x_i\mid\vec x_{[i-1]})$ denote the \emph{objective} of agent $i$. We call the first constraint the return on spend (RoS) constraint. The second, we refer to as the budget constraint. Note that given a pricing rule $\vec p$, once $v_i,B_i,\t_i$ are sampled for all agents, the outcome profile is deterministic given an arbitrary tie-breaking rule. We use $\vec z(\vec v,\vec B,\vec \t)$ to denote this outcome profile, and omit dependence on $\vec p$ as it will be clear from context. Moreover, any bundle $x_i\in X_i$ satisfying the RoS and budget constraint is called a \emph{feasible bundle}. To define prices, our mechanisms frequently use an (offline) \emph{outcome rule} or \emph{allocation rule}, which is a function $\ALG$ that maps a realization $(\vec v,\vec B,\vec \t)$ to a feasible outcome profile.

\paragraph{Valuation Functions for Combinatorial Auctions}
For combinatorial auctions, we assume that each bidder $i$ has a \emph{monotone} valuation function $v_i: 2^{M} \to \mathbb{R}_{\ge 0}$ with $v_i(\emptyset) = 0$. Let $v_i: 2^M \rightarrow \mathbb{R}_{\ge 0}$ be a valuation function,
\begin{itemize}
\item $v_i$ is \emph{XOS} if there exists a class $\mathcal{L}_{i} = \{(v^\ell_{ij})_{j \in M} \in \mathbb{R}^m_{\ge 0}\}$ of additive valuations such that for every subset $S \subseteq M$, it holds that $v_i(S) = \max_{\ell \in \mathcal{L}_{i}} \sum_{j \in S} v^\ell_{ij}$. 
\item $v_i$ is subadditive if $v_i(S\cup T)\leq v_i(T)+v_i(S)$ for all $S,T\subseteq M$.
\end{itemize}
We remark that if $v_i(\cdot)$ is XOS or subadditive, then $\hat v_i(\cdot):=\min\{v_i(\cdot),B_i\}$ remains XOS or subadditive, respectively, for any $B_i\in (0,\infty]$.

\paragraph{Liquid Welfare Objective} The standard benchmark for measuring the efficiency of outcomes in budget-free settings (when $\vec B=\vec \infty$) is in terms of social welfare. For budgeted settings, we measure the efficiency of outcomes in terms of the \emph{liquid welfare} (following \cite{dobzinski2014efficiency}), which is also the standard benchmark in the autobidding literature, see \cite{aggarwal2024}. 
Intuitively, the liquid welfare of an allocation corresponds to the maximum amount of revenue that can be extracted from the agents. The liquid welfare of a feasible outcome profile $\vec x\in \mathcal F$ is determined by
$$\LW(\vec x):=\sum_{i\in [n]}\min\{v_i(x_i),B_i\}.$$
We shall use $\OPT(\vec v,\vec B, \mathcal H):=\argmax_{\vec x\in \mathcal H}\LW(\vec x)$ for the allocation that maximizes the liquid welfare w.r.t. the feasibility constraint $\mathcal H\subseteq X$ and omitting the dependence on $\mathcal H$ when it is clear from context. When it is not immediately clear in the context with respect to which valuations and budgets the liquid welfare is taken, we shall write $\LW_{(\vec v,\vec B)}(\vec x)$ as the liquid welfare under the allocation $\vec x$ with valuation and budget profiles $(\vec v,\vec B)$.
For an agent $i$ with valuation function $v_i$ and budget $B_i$, we use $\hat v_i(\cdot) := \min\{v_i(\cdot),B_i\}$ for the \emph{liquid valuation function}, or the \emph{budget-capped valuation,} of agent $i$. 
For budget-free settings, i.e., when $\vec B=\vec \infty$, we note that the social welfare function $\SW(\vec x):=\sum_{i\in [n]}v_i(x_i)$ coincides with the liquid welfare function $\LW(\vec x)$.

\paragraph{Prophet Inequalities} An instance of a \emph{combinatorial allocation problem} is defined by a distribution $\mathcal D$ and is called \emph{budget-free} if $\vec B$ is $\vec \infty$. The \emph{competitive ratio} of a pricing rule $\vec p$ w.r.t. a class of instances of online outcome problems $\mathcal C$, is then defined as
$$\comp(\mathcal C,\vec p)=\sup_{\mathcal D\in \mathcal C}\frac{\E_{(\vec v,\vec B,\vec \t)\sim \mathcal D}\left[\LW(\OPT(\vec v,\vec B))\right]}{\E_{(\vec v,\vec B,\vec \t)\sim \mathcal D}\left[\LW(\vec z(\vec v,\vec B,\vec \t))\right]},$$
where $\vec z(\vec v,\vec B,\vec \t)$ is the allocation induced by the posted price $\vec p$ on the realization $(\vec v,\vec B,\vec \t)\sim \mathcal D$.
In other words, it is the worst case ratio between the optimal offline liquid welfare and the liquid welfare achieved by the allocation induced by the posted price mechanism. Often, we shall also write $\OPT(\vtype)$ for simplicity's sake as it is independent of $\vec \t$ anyway.
We will consider instances where each agent's price sensitivity is contained in a set $\Sigma$  (chosen to be $\subseteq[0,1]$) and $E_{\vtype\sim \mathcal D}[\LW(\OPT)]<\infty$. Let $\mathcal I_{\Sigma}$ denote this instance class and let $\mathcal I_{\Sigma}^{\infty}$ denote its budget-free subclass.

\section{Approximation Guarantees}
\subsection{Budget-Reduction}
Consider an instance $\mathcal D$ with a pricing rule $\vec p$. We define a new distribution $\hat{\mathcal D}(\mathcal D,\vec p)$ where $(\hat{\vec v},\vec \infty,\hat{\vec \t})$ is sampled from $\hat{\mathcal D}(\mathcal D,\vec p)$ as follows. Sample $(\vec v,\vec B,\vec \t)\sim \mathcal D$, choose $\hat v_i(\cdot):=\min\{v_i(\cdot),B_i\}$ and partition the agents into two sets. Let $N_1=\{i\in [n]\mid v_i(z_i(\vec v,\vec B,\vec \t))\geq B_i\}$ and $N_2=[n]\setminus N_1$. For $i\in N_1$, we let $\hat \t_i=0$ while for $i\in N_2$ we let $\hat \t_i=\t_i$. We call this transformation of instances the \emph{budget reduction} and $\hat{\mathcal D}$ the \emph{reduced instance} or \emph{proxy instance}. We let $\hat\o_i(\vec x):=\hat v_i(x_i)-\hat\sigma_i p_i(x_i\mid \vec x_{[i-1]})$, $\hat u_i(\vec x):=\hat v_i(x_i)- p_i(x_i\mid \vec x_{[i-1]})$ and will often omit the dependence of $\mathcal D$ and $\vec p$ on $\hat{\mathcal D}$. Lastly, note that if $\mathcal D$ is budget-free, then $\hat{\mathcal D}=\mathcal D$. The following proposition is a useful behavioural reduction in the proxy setting.

\begin{restatable}{proposition}{propdeviation}\label{prop:deviation}
    Given an instance $\mathcal D$, a realization $(\vec v,\vec B,\vec \t)\sim \mathcal D$ with pricing rule $\vec p$, we have for every agent $i$ and any feasible bundle $y_i$ given $\vec z_{[i-1]}(\vec v,\vec B,\vec \t)$  the following inequality:
    \begin{equation}\label{eq:deviation}\hat \o_i(\vec z(\vec v,\vec B,\vec \t))\geq \hat v_i(y_i)-p_i(y_i\mid \vec z_{[i-1]}(\vec v,\vec B,\vec \t)).\end{equation}
\end{restatable}
\begin{proof}
    To simplify the notation, we abbreviate $\vec z := \vec z(\vec v,\vec B,\vec \t)$ and distinguish between two cases.

    \bigskip

    \noindent
    \textbf{Case 1: $v_i(z_i)\geq B_i$.} We immediately find that
    $$\hat \o_i(\vec z)=\hat v_i(z_i)-\hat \sigma_i \cdot p_i(z_i\mid \vec z_{[i-1]})=B_i\geq \hat v_i(y_i)- p_i(y_i\mid \vec z_{[i-1]}).$$
    The second equality follows because $\hat \sigma_i=0$ and by definition of $\hat v_i$. The inequality follows again by definition of $\hat v_i$ and the fact that $p_i(y_i\mid \vec z_{[i-1]})\geq 0$.

    \bigskip

    \noindent
    \textbf{Case 2: $v_i(z_i)< B_i$.} In this case, we argue
    $$\hat \o_i(\vec z)=\hat v_i(z_i)-\hat \sigma_i \cdot  p_i(z_i\mid \vec z_{[i-1]})=v_i(z_i)-\sigma_i \cdot  p_i(z_i\mid \vec z_{[i-1]}).$$
    Here, the second equality follows since $v_i(z_i)<B_i$ and thus $\hat \sigma_i=\sigma_i$. Now, if $p_i(y_i\mid \vec z_{[i-1]})\leq \min\{v_i(y_i),B_i\}$, then $v_i(z_i)-\sigma_i \cdot  p_i(z_i\mid \vec z_{[i-1]})\geq v_i(y_i)-\sigma_i \cdot p_i(y_i\mid \vec z_{[i-1]})$ by the definition of $\vec z$ being the bundle that maximizes the objective. Then, the inequality in Equation~\eqref{eq:deviation} readily follows by the definition of $\hat v_i$ and the fact that $\sigma_i\leq 1$. If $p_i(y_i\mid \vec z_{[i-1]})> \min\{v_i(y_i),B_i\}$, then the inequality from Equation~\eqref{eq:deviation} follows immediately as the RHS is negative.
\end{proof}

The reduction in Proposition~\ref{prop:deviation} says that the proxy objective of the chosen bundle from the original instance is lower bounded by the proxy utility of any feasible bundle. Thus, lower bound analyses for utility maximizers where agents deviate by buying a specifically chosen feasible bundle can be lifted through Proposition~\ref{prop:deviation} to hold more generally. Such arguments to lower bound the utility appear e.g. in \cite{feldman2014,Dutting2020,dutting2024}.

Let $\mathcal D$ be an instance and $\Sigma$ be the support of the marginal distribution over price sensitivities. Then, we define $\t_{\min}=\min \Sigma$ if $\mathcal D$ is budget-free and otherwise as $0$. The reason is that our budget reduction works via setting some agents' price sensitivity to 0, effectively equalizing $\sigma_{\min}$, and with it the resulting competitive ratios, to be worst-possible in the presence of budgets. This is a consequence of the necessity to match price's effect in reduced instances to that in original ones appropriately. 

\subsection{Utility-Bounded Prices}

A useful characterization of pricing rules that we study is captured in the following definition.

\begin{definition}\label{def:utility-bounded}
    Let $\vec p$ be a pricing rule and $\lambda,\mu\geq 0$. We say that a pricing rule $\vec p$ is $(\lambda,\mu)$-utility-bounded (UB) w.r.t. an allocation rule $\ALG$ if for every instance $\mathcal D$ and every realization $\vtype \sim \mathcal D$, there exists a distribution $Y_i^{\vtype}$ over feasible bundles available to agent i under $\vtype$
    \begin{equation}\label{eq:utility-bounded-bound}
        \E_{\substack{\vtype \sim \mathcal D\\ y_i^{\vtype}\sim Y_i^{\vtype}}}\left[\sum_{i\in N}\hat v_i(y_i^{\vtype})-p_i(y_i^{\vtype}\mid \vec z_{[i-1]}(\vtype))\right]\geq \lambda \E_{\vtype\sim \mathcal D}[\LW(\ALG(\vtype))]-\mu \rev.
    \end{equation}
Here, $\rev$ is the expected revenue that the pricing rule $\vec p$ attains over $\mathcal D$. If $\ALG=\OPT$, we shall just say that the pricing rule is $(\lambda,\mu)$-UB.
\end{definition}

Next, we prove that UB pricing rules have constant welfare guarantees.

\begin{theorem}\label{thm:utility-bounded-guarantee}
    Let $\mathcal D$ be an instance, $\lambda>0$ and $\mu\geq 0$. A pricing rule $\vec p$ that is $(\lambda,\mu)$-UB w.r.t. an allocation rule $\ALG$ achieves an expected liquid welfare of at least
\begin{equation}\label{eq:utility-bound-guarantee}\frac{\lambda}{\max\{1,1+\mu-\sigma_{\min}\}}
    \E_{\vtype\sim \mathcal D}[\LW(\ALG(\vtype))].
\end{equation}
\end{theorem}
\begin{proof}
        Since $\vec p$ is $(\lambda,\mu)$-UB, there exist distributions $Y_i^{\vtype}$ such that the inequality in Equation~\eqref{eq:utility-bounded-bound} holds. Since $Y_i^{\vtype}$ is a distribution over feasible bundles, we can apply Proposition~\ref{prop:deviation} to obtain the following bound:
    \begin{align*}
        \sum_{i\in N}\E_{\vtype\sim \mathcal D}[\hat \Phi_i(\vec z(\vtype))]&\geq  \E_{\substack{\vtype \sim \mathcal D\\ y_i^{\vtype}\sim Y_i^{\vtype}}}\left[\sum_{i\in N}\hat v_i(y_i^{\vtype})-p_i(y_i^{\vtype }\mid \vec z_{[i-1]}(\vtype))\right]\\
        &\geq \lambda \E_{\vtype \sim \mathcal D}[\LW(\ALG(\vtype))]-\mu \rev.
    \end{align*}
    Now, since $\sigma_{\min}\leq \hat \sigma_i$ for any $i$ in any realization, the following lower bound on the difference between the expected liquid welfare achieved by the pricing rule and revenue scaled by $\sigma_{\min}$ holds
    \begin{align*}\E_{\vtype\sim \mathcal D}[\LW(\vec z(\vtype))]-\sigma_{\min}\rev&=\sum_{i\in N}\E_{\vtype\sim D}\left[\hat v_i(z_i(\vtype))-\sigma_{\min}p_i(z_i(\vtype)\mid \vec z_{[i-1]}(\vtype))\right]\\
    &\geq \sum_{i\in N}\E_{\vtype\sim D}\left[\hat v_i(z_i(\vtype))-\hat \sigma_i p_i(z_i(\vtype)\mid \vec z_{[i-1]}(\vtype))\right]\\
    &= \sum_{i\in N}\E_{\vtype\sim \mathcal D}[\hat \Phi_i(\vec z(\vtype))]\\
    &\geq \lambda \E_{\vtype \sim \mathcal D}[\LW(\ALG(\vtype))]-\mu \rev.\end{align*}
    Rearranging terms then gives
    $$\E_{\vtype\sim \mathcal D}[\LW(\vec z(\vtype))]+(\mu-\sigma_{\min})\rev\geq \lambda \E_{\vtype\sim \mathcal D}[\LW(\ALG(\vtype))].$$
    If $\mu\leq \sigma_{\min}$, then the $\rev$ term is non-positive and we obtain the first branch ($1/\lambda$-competitive relative to $\ALG$). If $\mu>\sigma_{\min}$, then the $\rev$ term is positive. Using the RoS and budget constraint, the expected revenue is upper bounded by the expected liquid welfare achieved by the pricing rule $\vec p$. Hence, the second branch is proven as follows:
    \begin{align*}(1+\mu-\sigma_{\min})\E_{\vtype\sim \mathcal D}[\LW(\vec z(\vtype))]&=\E_{\vtype\sim \mathcal D}[\LW(\vec z(\vtype))]+(\mu-\sigma_{\min})\E_{\vtype\sim \mathcal D}[\LW(\vec z(\vtype))]\\
    &\geq \E_{\vtype\sim \mathcal D}[\LW(\vec z(\vtype))]+(\mu-\sigma_{\min})\rev\\
    &\geq \lambda \E_{\vtype\sim \mathcal D}[\LW(\ALG(\vtype))].
    \end{align*}
\end{proof}

\begin{corollary}\label{cor:simple-bound}
    Suppose that a $(\lambda,\mu)$-UB pricing rule $\vec p$ is $C:=\frac{\max\{1,\mu\}}{\lambda}$-competitive for budget-free settings with utility maximizers. Then, in the general model, the competitive ratio deteriorates to at most $C(2-\sigma_{\min})$.
\end{corollary}
\begin{proof}
    Using Theorem~\ref{thm:utility-bounded-guarantee}, we see that in the general model the competitive guarantee deteriorates to $\frac{\max\{1,1+\mu-\sigma_{\min}\}}{\lambda}$. Dividing by $C$ gives the following expression:
    $$\frac{\max\{1,1+\mu-\sigma_{\min}\}}{\max\{1,\mu\}}.$$
    Now, this expression as a function in $\mu$ attains its maximum at $\mu=1$. Plugging this in, we immediately see that this ratio becomes $2-\sigma_{\min}$.
\end{proof}

From Corollary~\ref{cor:simple-bound}, it becomes immediately clear  that for $(\lambda,\mu)$-utility bounded pricing rules, budgets only worsen our competitive guarantees by a factor of at most 2.

\section{Combinatorial Auctions}\label{sec:combinatorial-auctions}
\subsection{XOS Combinatorial Auctions}
\paragraph{Example} As an example, we shall now prove that the half-pricing rule for XOS combinatorial auctions is $(\frac12,1)$-UB. More specifically, this pricing rule is defined as follows. For an XOS valuation and budget profile $(\vec v,\vec B)$, let $\vec{\tilde v}$ denote the additive support of $\vec{\hat v}$ for the allocation $\OPT(\vec v,\vec B)$\footnote{Note that we assume $\OPT$ to be unique, if not, we pick the first optimal allocation in lexicographic order: e.g. via a fixed ordering of agents and items.}. Then, for each item $j\in M$, we define
\begin{equation}\label{eq:pricing-rule-XOS}p^{(\vec v,\vec B)}(\{j\})=\sum_{\substack{i\in [n] \; : \;j\in \OPT_i(\vec v,\vec B)}}\tilde v_{ij},\end{equation}
and extend linearly to bundles. If an item is sold, its price is set to $\infty$. In other words, for the half-pricing rule, $\vec p_{\frac12}$, we let $p_j:=\frac12 \E_{(\vec v,\vec B,\vec \t)\sim \mathcal D}[p^{(\vec v,\vec B)}(\{j\})]$.

\begin{lemma}\label{lem:half-price}
    The half-pricing rule in XOS combinatorial auctions, $\vec p_{\frac12}$, is $(\frac12,1)$-UB.
\end{lemma}
\begin{proof}
    Consider an instance $\mathcal D$ and let $\vtype\sim \mathcal D$. Let $R_i(\vtype):=M\setminus \bigcup_{k<i}z_k(\vec \tau)$ be the set of remaining items when agent $i$ arrives. Sample a profile $\vtype_{-i}'$ for all other agents and let $\tilde v_i(\tau_i,\vtype_{-i}')$ be the additive support of the capped valuation $\hat v_i$ for the allocation $\OPT_i(\tau_i,\vtype_{-i}')$. Now, define the deviating bundle
    $$y_i^{\vtype}(\vtype_{-i}')=R_i(\vtype)\cap \OPT_i(\tau_i,\vtype_{-i}')\cap \{j\in M\mid \tilde v_{ij}(\tau_i, \vtype_{-i}')\geq  p_j\}.$$

    Notice that $y_i^{\vtype}$ is RoS and budget-feasible by construction and the definition of XOS valuations. Now, if we let $q_j$ denote the probability that item $j$ is sold, then we can lower bound the proxy utility as follows:
    {\allowdisplaybreaks
    \begin{align*}\E_{\vtype,\vtype'\sim \mathcal D}[\hat v_i(y_i^{\vtype}(\vtype_{-i}'))-p_i(y_i^{\vtype}(\vtype_{-i}'))]
    &\geq \E_{\vtype,\vtype'\sim \mathcal D}\left[\sum_{j\in y_i^{\vtype}(\vtype_{-i}')}\tilde v_{ij}(\tau_i,\vtype_{-i}')-p_{j}\right]\\
    &= \E_{\vtype,\vtype'\sim \mathcal D}\left[\sum_{j\in \OPT_i(\tau_i,\vtype_{-i}')}(\tilde v_{ij}(\tau_i,\vtype_{-i}')-p_{j})^+\mathbbm{1}\{j\in R_i(\vtype)\}\right]\\
    &=\E_{\vtype,\vtype'\sim \mathcal D}\left[\sum_{j\in \OPT_i(\vtype)}(\tilde v_{ij}(\vtype)-p_{j})^+\mathbbm{1}\{j\in R_i(\tau_i, \vtype_{-i}')\}\right]\\
    &\geq \E_{\vtype\sim \mathcal D}\left[\sum_{j\in \OPT_i(\vtype)}(1-q_j)(\tilde v_{ij}(\vtype)-p_{j})^+\right]\\
    &\geq \E_{\vtype\sim \mathcal D}\left[\sum_{j\in \OPT_i(\vtype)}(1-q_j)(\tilde v_{ij}(\vtype)-p_{j})\right],
    \end{align*}}
    where $(t)^+=\max\{0,t\}$. The first inequality follows by definition of XOS. For the first equality, note that for $j\in \OPT_i(\tau_i,\vtype_{-i}')$, if $\tilde v_{ij}(\tau_i,\vtype_{-i'})< p_j$, such terms contribute 0. The second equality follows from $\vtype$ and $\vtype'$ being identically distributed and $R_i(\cdot)$'s independence of agent $i$ or any of the later agents. The second  inequality uses that the probability the item is not sold is at most the probability the item is still available when agent $i$ arrives. Summing over all agents and using the fact that each item price is exactly half its expected contribution to the liquid welfare then gives
    $$\E_{\vtype,\vtype'\sim \mathcal D}\left[\sum_{i\in [n]}\hat u_i(y_i^{\vtype}(\vtype_{-i}'))\right]\geq \frac12 \E_{\vtype\sim \mathcal D}[\LW(\OPT(\vtype))]-\sum_{j\in M}q_jp_j.$$
    Note that $\sum_{j\in M}q_jp_j$ is exactly the revenue, proving $(\frac12,1)$-utility-boundedness.
\end{proof}
\paragraph{General Upper Bounds} Extending the proof of Lemma~\ref{lem:half-price} to arbitrary $\delta\in (0,1)$ we easily see that the pricing rule $\vec p_{\delta}$ with $p_j:=\delta \E_{(\vec v,\vec B,\vec \t)\sim \mathcal D}[p^{(\vec v,\vec B)}(\{j\})]$ is $(1-\delta,\frac{1-\delta}{\delta})$-UB.
The competitive ratio of the pricing rule $\vec p_{\delta}$ can be upper bounded by Theorem~\ref{thm:utility-bounded-guarantee}. Using this, we obtain for budget-free instances the following bound
\begin{equation}\label{eq:UB-combinatorial}\comp(\mathcal I^{\infty}_{\Sigma},\vec p_{\delta})\leq \begin{cases}\frac{1-\delta \t_{\min}}{(1-\delta)\delta}, &\text{if $\delta\leq \frac{1}{1+\t_{\min}}$}\\
\frac{1}{1-\delta}, & \text{if $\delta >\frac{1}{1+\t_{\min}}$}.
\end{cases}\end{equation}
For $\delta=\frac12$ and $\Sigma=\{0\}$ we recover the bound of 4 found in \cite{deng2022} for value maximizers.

When agents are constrained by budgets, we simply set $\sigma_{\min}=0$ in Equation~\eqref{eq:UB-combinatorial} and obtain
\begin{equation}\label{eq:UB-budgets-combinatorial}\comp(\mathcal I_{\Sigma},\vec p_{\delta})\leq \frac{1}{(1-\delta)\delta}.\end{equation}
This bound is smallest-possible for $\delta=\frac12$, yielding again 4. For utility maximizers, i.e., $\Sigma=\{1\}$, the same ratio was also obtained in \cite{fotakis2019}.

\paragraph{Lower bounds.}
The bounds in Equations~\eqref{eq:UB-combinatorial}~\&~\eqref{eq:UB-budgets-combinatorial} are also tight. When agents are budget-free and valuations are additive, the price for item $j$ under $\vec p_{\delta}$ reduces to the following form that will be useful in our lower bound constructions
$$p_{\delta}(\{j\})=\delta \E_{(\vec v,\vec \t)\sim \mathcal D}\left[\max_{i\in [n]}v_{ij}\right],$$
note that this does not depend on $\vec \t$. 

\begin{restatable}{lemma}{budgetfreeLB}
        Let $\sigma_{\min}=\min \Sigma$. Then, for budget-free instances, we have that 
    $$\comp(\mathcal I^{\infty}_{\Sigma},\vec p_\delta)\geq \max\left\{\frac{1-\delta\cdot \sigma_{\min}}{(1-\delta)\delta},\frac{1}{1-\delta}\right\}.$$
\end{restatable}
\begin{proof}\label{proof:lemma-budget-free-LB}
    Consider $m=2k$ items and three agents with additive valuations. Let the price sensitivities of these agents be drawn from any distribution with support in $\Sigma$, but let agent 1 have a deterministic price sensitivity of $\sigma_1=\sigma_{\min}$. Agent 1 has a value of $\left(\delta\cdot k+\frac{\delta^2}{1-\delta}\right)(1-\sigma_{\min})$ for exactly one of the first $k$ items, each with a uniform probability. The last $k$ items they value at $v_{1j}=\sigma_{\min}\cdot \delta+\epsilon$ where $\epsilon>0$ such that $v_{1j}<1$ for $j>k$. Agent 2 only values the first $k$ items at $v_{2j}=(1-\sigma_{\min})\frac{\delta^2}{1-\delta}-\eta$ for $1\leq j \leq k$ for some $(1-\sigma_{\min})\frac{\delta^2}{1-\delta}>\eta>0$. Note that such an $\eta$ exists whenever $\sigma_{\min}<1$. We treat this case separately below. The last agent values only the last $k$ items at $v_{3j}=1$ for $j>k$.

    In this instance, optimally, exactly one of the first $k$ items is allocated to agent 1, the other $k-1$ of the first $k$ items are allocated to agent 2 while the last $k$ items are allocated to agent 3. The liquid welfare generated by this assignment is constant over all realizations and is
    \begin{align*}\OPT&=\left(\delta\cdot k+\frac{\delta^2}{1-\delta}\right)(1-\sigma_{\min})+(k-1)\left((1-\sigma_{\min})\frac{\delta^2}{1-\delta}-\eta\right)+k\\
    &=k\left(\delta+\frac{\delta^2}{1-\delta}\right)(1-\sigma_{\min})+k-(k-1)\eta.\\
    &=k\cdot \frac{\delta}{1-\delta}\cdot (1-\sigma_{\min})+k-(k-1)\eta
    \end{align*}

    Now, to determine the generated social welfare by the pricing rule, let us first determine the prices of each item. For the last $k$ items, we see that the price is easily $p_j=\delta$ since the maximum value for these items is 1 and agent 3 has additive valuations. For the first $k$ items, we have that with probability $\frac1k$ that item $j$ is allocated to agent 1 under the optimal allocation. Otherwise, it should be allocated to agent 2. Thus, we see that for $1\leq j \leq k$
    \begin{align*}p_j&=\frac{\delta}{k}\left(\delta\cdot k+\frac{\delta^2}{1-\delta}\right)(1-\sigma_{\min})+\delta\cdot \left(1-\frac1k\right)\cdot \left((1-\sigma_{\min})\frac{\delta^2}{1-\delta}-\eta\right)\\
    &=\left(\delta^2+\frac{\delta^3}{1-\delta}\right)(1-\sigma_{\min})-\delta\cdot \eta \cdot \left(1-\frac1k\right)\\
    &=(1-\sigma_{\min})\frac{\delta^2}{1-\delta}-\delta\cdot \eta \cdot \left(1-\frac1k\right)\\
    &>(1-\sigma_{\min})\frac{\delta^2}{1-\delta}- \eta\\
    &=v_{2j}.
    \end{align*}
    Hence, agent 2 cannot buy these items. From the first $k$ items, agent 1 will buy the item for which they have a non-zero value and also buy the last $k$ items because for $j>k$ we have $v_{1j}= \sigma_1\cdot \delta+\epsilon > \sigma_1\cdot p_j$. In other words, buying these items contributes to their objective. The payment for buying this bundle is
    $$(1-\sigma_{\min})\frac{\delta^2}{1-\delta}-\delta\cdot \eta \cdot \left(1-\frac1k\right)+\delta\cdot k.$$
    The value they attain from this bundle is
    $$\left(\delta\cdot k+\frac{\delta^2}{1-\delta}\right)(1-\sigma_{\min})+ (\sigma_{\min}\cdot \delta+\epsilon)\cdot k=(1-\sigma_{\min})\frac{\delta^2}{1-\delta}+\delta\cdot k+\epsilon \cdot k.$$
    Hence, buying this bundle is feasible for agent 1. If agent 1 is a value maximizer, $\sigma_1=0$, then they might choose to buy more items from the first $k$ items, but this generates no welfare. Thus, the social welfare generated by the pricing rule, $\SW$, is also constant over all realizations and given by
    $$\SW=(1-\sigma_{\min})\frac{\delta^2}{1-\delta}+\delta\cdot k+\epsilon\cdot k.$$
    Taking $\epsilon,\eta\to 0$ and $k\to \infty$ then gives the following ratio
    $$\lim_{k\to \infty}\lim_{\eta\to 0}\lim_{\epsilon\to 0}\frac{\OPT}{\SW}=\lim_{k\to \infty}\frac{k\cdot \frac{\delta}{1-\delta}\cdot (1-\sigma_{\min})+k}{(1-\sigma_{\min})\frac{\delta^2}{1-\delta}+\delta\cdot k}=\frac{ \frac{\delta}{1-\delta}\cdot (1-\sigma_{\min})+1}{\delta}=\frac{1-\delta\cdot \sigma_{\min}}{\delta(1-\delta)}.$$

    Next, suppose that $\t_{\min}=1$. Consider an instance with two agents and one item. Agent 1 has a value of $\delta+\epsilon$, such that $\delta+\epsilon<1$, while agent 2 has a value of $1$. In that case, $\E[\OPT]=1$, so the price becomes $\delta$. Hence, agent 1 buys the item. In this case, the ratio becomes exactly $\frac{1}{\delta+\epsilon}$. Taking the limit as $\epsilon\to 0$ then gives $\frac1\delta$. This coincides with $\frac{1-\delta\cdot \t_{\min}}{\delta(1-\delta)}$ when $\t_{\min}=1$.

    Lastly, consider the instance with one item, one agent, and arbitrary $\sigma$. The valuation of agent 1 is given $\frac1\epsilon$ with probability $\epsilon$ and otherwise $\frac{\delta}{1-\delta(1-\epsilon)}-\eta$. The price of this instance is given by
    $$\delta+\delta\cdot \left(\frac{\delta}{1-\delta(1-\epsilon)}-\eta\right)\cdot (1-\epsilon)=\frac{\delta}{1-\delta(1-\epsilon)}-\delta\cdot \eta\cdot (1-\epsilon)>\frac{\delta}{1-\delta(1-\epsilon)}-\eta.$$
    Thus, this agent cannot buy when their value is $\frac{\delta}{1-\delta(1-\epsilon)}-\eta$. Hence, the expected social welfare generated by the algorithm is 1. Taking $\eta\to 0$ and then $\epsilon\to 0$ gives the following ratio
    $$\lim_{\epsilon\to0}\lim_{\eta\to 0}\frac{\E[\OPT]}{\E[\SW]}=\frac{1}{1-\delta},$$
    as desired.
\end{proof}

Interestingly, the construction only makes use of stochastic additive valuations while the agents' objectives are deterministic. However, the upper bound holds for stochastic agent objectives and XOS-valuations, proving that stochastic objectives do not deteriorate the welfare guarantees.
For budgeted instances, it turns out that using additive valuations, deterministic objectives and deterministic budgets suffices to match the more general upper bound in Equation~\eqref{eq:UB-budgets-combinatorial}.
\begin{restatable}{lemma}{budgetLB}\label{lemma:budget-LB}
For any $\Sigma$ we have $\comp(\mathcal I_{\Sigma},\vec p_\delta)\geq \frac{1}{\delta(1-\delta)}$.  
\end{restatable}
\begin{proof}
    Consider $m=2k$ number of items and three agents with additive valuations. Let the price sensitivities of these agents be drawn from any distribution with support contained in $\Sigma$. Agent 1 has a budget of $B_1=\delta\cdot k+\frac{\delta^2}{1-\delta}$ and has a value of $B_1$ for exactly one of the first $k$ items each with uniform probability. The last $k$ items they value at $v_{1j}=\delta+\epsilon$ with $\epsilon >0$ so that $v_{1j}<1$ for $j>k$. Agent 2 is budget-free and only values the first $k$ items at $v_{2j}=\frac{\delta^2}{1-\delta}-\eta$ for $1\leq j \leq k$ for some $\frac{\delta^2}{1-\delta}>\eta>0$. The last agent values only the last $k$ items at $v_{3j}=1$ for $j>k$ and is also budget-free.

    In this instance, optimally, exactly one of the first $k$ items is allocated to agent 1, the other $k-1$ of the first $k$ items are allocated to agent 2 while the last $k$ items are allocated to agent 3. The liquid welfare generated by this assignment is constant over all realizations and is
    $$\OPT=\delta \cdot k+\frac{\delta^2}{1-\delta}+(k-1)\left(\frac{\delta^2}{1-\delta}-\eta\right)+k=\delta\cdot k+k\cdot \frac{\delta^2}{1-\delta}+k-\eta\cdot (k-1).$$

    Now, to determine $\LW$, the liquid welfare generated by the pricing rule, let us first determine the prices of each item. For the last $k$ items, we see that the price is easily $p_j=\delta$ as the maximum value for these items is 1 and agent 3 has additive valuations. For the first $k$ items, we have that with probability $\frac1k$ that item $j$ is allocated to agent 1 under the optimal allocation. Otherwise, it should be allocated to agent 2. Thus, we see that for $1\leq j \leq k$ that
    \begin{align*}p_j&=\frac{\delta}{k}\left(\delta\cdot k + \frac{\delta^2}{1-\delta}\right)+\delta\cdot \left(1-\frac1k\right)\cdot \left(\frac{\delta^2}{1-\delta}-\eta\right)\\
    &=\delta^2+\frac{\delta^3}{1-\delta}-\delta\cdot \eta \cdot \left(1-\frac1k\right)\\
    &=\frac{\delta^2}{1-\delta}-\delta\cdot \eta \cdot \left(1-\frac1k\right)\\
    &>\frac{\delta^2}{1-\delta}- \eta\\
    &=v_{2j}.
    \end{align*}
    Hence, agent 2 cannot buy these items. Agent 1, will buy definitely buy one item from the first $k$ items and then also buy the last $k$ items because for $j>k$ we have $v_{1j}=\delta+\epsilon  > \sigma_1\cdot \delta = \sigma_1\cdot p_j$. In other words, agent 1 buys these items. The payment for buying this bundle is
    $\frac{\delta^2}{1-\delta}-\delta\cdot \eta \cdot \left(1-\frac1k\right)+\delta\cdot k$
    which is at most $B_1$. Hence, buying this bundle is budget-feasible for agent 1. If agent 1 is a value maximizer, $\sigma_1=0$, then they might choose to buy more items from the first $k$ items, but this generates no welfare. Hence, the liquid welfare generated by the pricing rule, $\LW$, is also constant over all realizations, and we see that
    $\LW=\delta\cdot k+\frac{\delta^2}{1-\delta}$.
    Taking $\eta\to 0$ and $k\to \infty$ then gives the following ratio
    $$\lim_{k\to \infty}\lim_{\eta\to 0}\frac{\OPT}{\LW}=\lim_{k\to \infty}\frac{\delta\cdot k+k\cdot \frac{\delta^2}{1-\delta}+k}{\delta\cdot k+\frac{\delta^2}{1-\delta}}=\frac{\delta+1+\frac{\delta^2}{1-\delta}}{\delta}=\frac{1}{\delta(1-\delta)}.$$    
\end{proof}

These lower bounds show that our upper bounds for XOS combinatorial auctions are tight, and thus also settle tightness of the upper bounds found for utility maximizers with budgets in \cite{fotakis2019} and budget-free value maximizers in \cite{deng2022} for the half-pricing rule $\vec p_{\frac12}$.

\paragraph{Optimal Factors}
When $\delta=\frac12$, we call $\vec p_{\frac12}$ the \emph{half-pricing rule}. If $\sigma_{\min}=0$ or $\sigma_{\min}=1$ this turns out to be the optimal choice for $\delta$. Fixing $\delta=\frac12$, we obtain $\comp(\mathcal I^{\infty}_{\Sigma},\vec p_{\frac12})=4-2\sigma_{\min}$, while with budgets, the competitive ratio remains $4$ for all $\Sigma$.
However, the half-pricing rule is never optimal for any $\sigma_{\min}\in (0,1)$ -
we state, omitting the (simple calculus) proof, the optimal $\delta$ in terms of $\t_{\min}$ in the following proposition.
\begin{proposition}
    For budget-free combinatorial auctions with XOS valuations the following choices for $\delta$ are optimal and give the corresponding competitive guarantees:
    $$\delta^\ast = \begin{cases}
        \frac{1}{1+\sqrt{1-\t_{\min}}}, & \text{if $0\leq \t_{\min}<\frac{\sqrt{5}-1}{2}$}\\
        \frac{1}{1+\t_{\min}}, & \text{if $\frac{\sqrt{5}-1}{2}\leq \t_{\min}\leq 1$}
    \end{cases}
    \qquad \comp(\mathcal I_{\Sigma}^{\infty},\vec p_{\delta^\ast}) = \begin{cases}
        (1+\sqrt{1-\t_{\min}})^2, & \text{if $0\leq \t_{\min}<\frac{\sqrt{5}-1}{2}$}\\
        \frac{1+\t_{\min}}{\t_{\min}}, & \text{if $\frac{\sqrt{5}-1}{2}\leq \t_{\min}\leq 1$.}
    \end{cases}$$
\end{proposition}

\subsection{Subadditive Combinatorial Auctions}
A breakthrough result by \cite{correa2023constant} recently showed existence of a $6$-competitive prophet inequality for the more general class of combinatorial auctions with subadditive valuations, establishing the following theorem.

\begin{theorem}[Corollary 1 \cite{correa2023constant}]\label{thm:constantfactorsubadditive} Given an instance with subadditive valuations and fixed arrival order of the agents, there exists an online algorithm, $\ALG^\ast$, such that
$$6\E_{\vec v\sim \mathcal D}[\SW(\ALG^\ast(\vec v))]\geq \E_{\vec v\sim \mathcal D}[\SW(\OPT(\vec v))].$$
\end{theorem}

In this section, we briefly explain the existence of personalized bundle prices that are $6(2-\sigma_{\min})$-competitive for subadditive agents with stochastic objectives and budgets. Let $L_i(X)$ denote the expected liquid welfare the optimal online algorithm attains from agent $i$ and onward given that the set of remaining items is $X$ upon agent $i$'s arrival. In that case, we have that $L_1(M)$ is the total expected liquid welfare the online optimal algorithm achieves and that for any $X\subseteq M$
\begin{align*}L_n(X)&=\E_{(v_n,B_n,\t_n)\sim \mathcal D_n}[\hat v_n(X)]\\
L_i(X)&=\E_{(v_i,B_i,\t_i)\sim \mathcal D_i}\left[\max_{Y\subseteq X}\left(\hat v_i(Y)+L_{i+1}(X\setminus Y)\right)\right].\end{align*}

Now, consider the dynamic personalized bundle price where if $R$ is the set of remaining items when agent $i$ arrives, the price for each $X\subseteq R$ is defined by:
\begin{equation}\label{eq:prices-subadditive}p_{i,R}(X):=L_{i+1}(R)-L_{i+1}(R\setminus X),\end{equation}
where we define $L_{n+1}(X):=0$ for all $X\subseteq M$.

In the budget-free setting with only utility maximizers, these prices are known to implement the online optimal algorithm. However, for arbitrary objectives and budgets this need not be true. We can, however, prove the following lemma.

\begin{lemma}\label{lemma:UB-subadditive}
    The pricing rule defined by the dynamic personalized bundle prices defined by Equation~\eqref{eq:prices-subadditive} is $(1,1)$-UB w.r.t. the online optimal algorithm.
\end{lemma}
\begin{proof}
    Let $R_i(\vtype)=M\setminus \bigcup_{k<i}z_k(\vtype)$ be the set of remaining items when agent $i$ arrives and let $$y_i^{\vtype}\in \argmax_{X\subseteq R_i(\vtype)}\{\hat v_i(X)-p_{i,R_i(\vtype)}(X)\}.$$
    In other words, $y_i^{\vtype}$ is a proxy utility maximizing bundle. Since the payment for buying nothing is 0, this bundle is budget and RoS feasible. Let $Y_i^{\vtype}$ then be a degenerate distribution that is $y_i^{\vtype}$. Hence, in a realization $\vtype$, we see by the definition of the pricing rule in Equation~\eqref{eq:prices-subadditive} that the proxy utility can be rewritten as:
    $$\hat v_i(y_i^{\vtype})-p_{i,R_i(\vtype)}(y_i^{\vtype})=\max_{X\subseteq R_i(\vtype)}\{\hat v_i(X)-p_{i,R_i(\vtype)}(X)\}=\max_{X\subseteq R_i(\vtype)}\{\hat v_i(X)+L_{i+1}(R_i(\vtype)\setminus X)\}-L_{i+1}(R_i(\vtype)).$$
    Now, taking expectations, summing over all agents and using the definition of $L_i(R_i(\vtype))$, we see that
    $$\begin{aligned}\sum_{i\in [n]}\E_{\vtype\sim \mathcal D}[\hat v_i(y_i^{\vtype})-p_{i,R_i(\vtype)}(y_i^{\vtype})]&=\sum_{i\in [n]}\E_{\vtype\sim \mathcal D}[L_i(R_i(\vtype))-L_{i+1}(R_i(\vtype))]\\
    &=L_1(M)-\sum_{i=1}^{n-1}\E_{\vtype\sim \mathcal D}[L_{i+1}(R_{i}(\vtype))-L_{i+1}(R_{i+1}(\vtype))]\\
    &=L_1(M)-\sum_{i=1}^{n-1}\E_{\vtype\sim \mathcal D}[L_{i+1}(R_{i}(\vtype))-L_{i+1}(R_{i}(\vtype)\setminus z_i(\vtype))]\\
    &=L_1(M)-\sum_{i=1}^{n-1}\E_{\vtype\sim \mathcal D}[p_{i,R_i(\vtype)}(z_i(\vtype))]\\
    &=\E_{\vtype\sim \mathcal D}[\LW(\OPT_{\operatorname{online}}(\vtype))]-\rev,
    \end{aligned}$$
    where we used $\OPT_{\operatorname{online}}$ to denote the optimal online allocation algorithm.
    The last equality follows because $p_{n,R}(X)=0$ for any $X\subseteq R$ and any $R\subseteq M$.
\end{proof}

Note that if $v(\cdot)$ is subadditive, then $\min(v(\cdot),B_i)$ will also be subadditive. Hence, Theorem~\ref{thm:constantfactorsubadditive} extends to general model with stochastic budgets and the liquid welfare benchmark. Lemma~\ref{lemma:UB-subadditive} in combination with Theorem~\ref{thm:utility-bounded-guarantee}~\&~\ref{thm:constantfactorsubadditive} then show that the prices defined by equation~\eqref{eq:prices-subadditive} are $6(2-\sigma_{\min})$-competitive. The following theorem is thus immediate.

\begin{theorem}
    There exists a dynamic personalized bundle price rule for combinatorial auctions with subadditive agents so that its competitive ratio is at most $6(2-\sigma_{\min})$.
\end{theorem}

\section{Balanced Prices}
Following the framework introduced by \cite{Dutting2020}, a set of outcome profiles $\mathcal H\subseteq X$ is \emph{exchange-compatible} with $\vec x \in \mathcal F$ if for all $\vec y \in \mathcal H$ and all $i\in [n]$, $(y_i,\vec x_{-i})\in \mathcal F$. A family of sets $(\mathcal F_{\vec x})_{\vec x\in X}$ is called exchange-compatible if $\mathcal F_{\vec x}$ is exchange-compatible with $\vec x$ for all $\vec x\in X$. Now we introduce the definition of balanced prices for the budget-free setting.
\begin{definition}[Balanced Prices] \label{def:balancedprices}
       Let $\alpha>0$, $\beta_1,\beta_2\geq0$. Given a set of feasible outcomes $\mathcal F$, a valuation profile $\vec v$ and price sensitivity profile $\vec \t$, a pricing rule $\vec p$ is $(\alpha,\beta_1,\beta_2)$-balanced w.r.t. the allocation rule $\ALG$, an exchange-compatible family of sets $(\mathcal F_{\vec x})_{\vec x\in X}$, and an indexing of the agents $i=1,\dots,n$, if for all $\vec x\in \mathcal F$,
    \begin{enumerate}[label = (\alph*)]
        \item \begin{equation}\label{eq:high-enough}\sum_{i\in [n]}p_i(x_i\mid\vec x_{[i-1]})\geq \frac1\alpha\left(\SW(\ALG(\vec v,\vec \t))-\SW(\OPT(\vec v,\mathcal F_{\vec x}))\right)\end{equation}
        \item for all $\vec x'\in \mathcal F_{\vec x}$: 
        \begin{equation}\label{eq:low-enough}\sum_{i\in [n]}p_i(x_i'\mid\vec x_{[i-1]})\leq \beta_1\SW(\OPT(\vec v,\mathcal F_{\vec x}))+\beta_2 \SW(\ALG(\vec v,\vec \t)).\end{equation}
    \end{enumerate}
    Whenever $\beta_2=0$, we suppress the coordinate and say that the pricing rule is $(\alpha,\beta_1)$-balanced. Given a set of valuation and price sensitivity profiles $V$ and $S$, respectively, a collection of pricing rules $(\vec p^{(\vec v,\vec \t)})_{(\vec v,\vec \t)\in V\times S}$ is $(\alpha,\beta_1,\beta_2)$-balanced if there exists an exchange-compatible family of sets $(\mathcal F_{\vec x})_{\vec x\in X}$ such that for all $(\vec v,\vec \t)\in V\times S$, the pricing rule $\vec p^{(\vec v,\t)}$ is $(\alpha,\beta_1,\beta_2)$-balanced with respect to $(\mathcal F_{\vec x})_{\vec x\in X}$.
\end{definition}

Originally, the notion of balanced in Definition~\ref{def:balancedprices} was called weakly-balanced by \cite{Dutting2020}. We omit the distinction between weakly balanced prices and balanced prices.

Next, observe that the balanced prices may depend on the agent's objective. In other words, we could charge certain agents more or less depending on their objective. Somewhat surprisingly, to obtain prophet inequalities for budgeted settings in our model, it suffices to obtain prophet inequalities in budget-free settings with some additional mild assumptions. For this reason, balanced prices are only defined for budget-free settings. The following definitions state these mild assumptions.

\begin{definition}\label{def:preserved}
    An allocation rule $\ALG$ is said to be \emph{preserved} under the budget reduction if for any profile $(\vec v,\vec B,\vec \t)$ and its reduction $(\vec{\hat v},\vec \infty,\vec{\hat \t})$ we have that $\LW_{(\vec v,\vec B)}(\ALG(\vec v,\vec B,\vec \t))=\SW_{\vec{\hat v}}(\ALG(\vec{\hat v},\vec \infty,\vec{\hat \t}))$. Similarly, a pricing rule $\vec p$ that depends on $\mathcal D$ is said to be preserved under the budget reduction if $\vec p(\mathcal D)=\vec p(\hat{\mathcal D}(\mathcal D,\vec q))$, for any pricing rule $\vec q$.
\end{definition}

Note that the optimal allocation, $\OPT$, is preserved under the budget reduction. Suppose that a pricing rule $\vec p$ depends on $\mathcal D$ through $p_i(x_i\mid \vec y):=\E_{(\vec v,\vec B,\vec \t)\sim \mathcal D}\left[p^{(\vec v,\vec B,\vec \t)}_i(x_i\mid \vec y)\right]$ where $\vec p^{(\vec v,\vec B,\vec \t)}$ is a pricing rule that is defined for any tuple of profiles $(\vec v,\vec B,\vec \t)$. Then it is easy to see that a sufficient condition for $\vec p$ to be preserved under the budget reduction is $\vec p^{(\vec v,\vec B,\vec \t)}=\vec p^{(\vec {\hat v},\vec \infty,\vec {\hat \t})}$. So, for example, the pricing rule $\vec p_{\delta}$ for XOS combinatorial auctions is preserved under the reduction since it only depends on $(\vec v,\vec B)$ through $\hat{\vec v}$. We state the following lemma relating balanced prices to utility boundedness, following the original structure from \cite{Dutting2020}.

\begin{theorem}\label{lem:UB-balancedprices}
    Let $\vec p^{(\vec v,\vec B,\vec \t)}$ be a pricing rule defined for any $(\vec v,\vec B,\vec \t)$. Consider an instance $\mathcal D$ and the pricing rule $\vec p$, defined as
    $$p_i(x_i\mid \vec y):=\E_{(\vec v,\vec B,\vec \t)\sim \mathcal D}[p_i^{(\vec v,\vec B,\vec \t)}(x_i\mid \vec y)],$$ 
    and let $\alpha>0$, $\beta_1,\beta_2\geq 0$. Suppose that for the collection of pricing rules $(\vec p^{(\vec v,\vec B,\vec \t)})_{(\vec v,\vec B,\vec \t)\in \supp \mathcal D}$ for feasible outcomes $\mathcal F$, the collection of pricing rules $(\vec p^{(\vec{\hat v},\vec \infty,\vec{\hat \t})})_{(\vec {\hat v},\vec{\hat \t})\in \supp \hat{ \mathcal D}}$ is $(\alpha,\beta_1,\beta_2)$-balanced w.r.t. an allocation rule $\ALG$ and an exchange-compatible family of sets $(\mathcal F_{\vec x})_{\vec x\in X}$. Here, $\hat{\mathcal D}$ is the proxy instance w.r.t. the pricing rule $\delta \vec p$. If $\ALG$ and the pricing rule $\vec p$, are both preserved under the budget reduction, then the pricing rule $\delta \vec p$ is $(1-\delta\beta_1-\delta\beta_2,\alpha\frac{1-\delta\beta_1}{\delta})$-UB w.r.t. $\ALG$ for $0< \delta$ and $\delta(\beta_1+\beta_2)<1$.
\end{theorem}
To prove this Lemma, we first start by generalizing the utility bound from the proof of Theorem 3.2. in \cite{Dutting2020} to obtain the \emph{proxy utility bound}. Next, we also generalize their revenue bound to budget-constrained settings. We combine these bounds to prove utility-boundedness. The proofs of the proxy utility and revenue bounds can be found in the appendix since they are of similar structure to the utility and revenue bound proofs in \cite{Dutting2020}. For ease of notation, we shall write 
$$O:=\E_{\substack{\vtype\sim \mathcal D\\ (\vec v',\vec B')\sim \pi}}\left[\LW_{(\vec v',\vec B')}(\OPT(\vec v',\vec B',\mathcal F_{\vec z(\vtype)}))\right] \text{ and}\; A := \E_{ \vtype\sim \mathcal D}\left[\LW(\ALG( \vtype))\right].$$

\begin{restatable}[Proxy Utility Bound]{lemma}{lemmaobjective}\label{lemma:objectivebound}
Suppose that the pricing rule $\vec p$ satisfies the conditions from Theorem~\ref{lem:UB-balancedprices}. Then, there a distribution $Y_i^{\vtype}$ over feasible remaining bundles such that
\begin{equation}
    \begin{aligned}\E_{\substack{\vtype\sim \mathcal D\\ y_i^{\vtype}\sim Y_i^{\vtype}}}\left[\sum_{i\in [n]}\hat v_i(y_i^{\vtype})-\delta p_i(y_i^{\vtype}\mid \vec z_{[i-1]}(\vtype))\right]&\geq (1-\delta \beta_1)O -\delta \beta_2A.
    \end{aligned}
\end{equation}
\end{restatable}

\begin{restatable}[Revenue Bound]{lemma}{lemmarevenue}\label{lemma:revenuebound}
    Suppose that the pricing rule $\vec p$ satisfies the conditions from Theorem~\ref{lem:UB-balancedprices}. Then, for the pricing rule $\delta \vec p$ the revenue can be bounded below as follows:
    $$\begin{aligned}\E_{\vtype\sim \mathcal D}\left[ \sum_{i\in [n]}\delta p_i(z_i(\vtype)\mid \vec z_{[i-1]}(\vtype))\right]&\geq \frac\delta\alpha A  - \frac\delta\alpha O.
    \end{aligned}$$
\end{restatable}

\begin{proof}[Proof of Theorem~\ref{lem:UB-balancedprices}]
    
    Let $\mu\geq 0$, then using the bounds from Lemmas~\ref{lemma:objectivebound}~\&~\ref{lemma:revenuebound} we find that
    $$\begin{aligned}\E_{\substack{\vtype\sim \mathcal D\\ y_i^{\vtype}\sim Y_i^{\vtype}}}\left[\sum_{i\in [n]}\hat v_i(y_i^{\vtype})-\delta p_i(y_i^{\vtype}\mid \vec z_{[i-1]}(\vtype))\right]+\mu\rev &\geq (1-\delta\beta_1)O-\delta\beta_2 A+\mu\left(\frac{\delta}{\alpha}A-\frac{\delta}{\alpha}O\right)\\
    &=\left(1-\delta\beta_1-\mu \frac{\delta}{\alpha}\right)O+\left(\mu\frac{\delta}{\alpha}-\delta\beta_2\right)A.
    \end{aligned}$$
    Now, choosing $\mu = \alpha \frac{1-\delta\beta_1}{\delta}$ we see that the $O$ term vanishes and we are left with
    $$\E_{\substack{\vtype\sim \mathcal D\\ y_i^{\vtype}\sim Y_i^{\vtype}}}\left[\sum_{i\in [n]}\hat v_i(y_i^{\vtype})-\delta p_i(y_i^{\vtype}\mid \vec z_{[i-1]}(\vtype))\right]+\left(\alpha \frac{1-\delta\beta_1}{\delta}\right)\rev\geq (1-\delta\beta_1-\delta\beta_2)\E_{ \vtype\sim \mathcal D}\left[\LW(\ALG( \vtype))\right],$$
    as desired.
\end{proof}

We illustrate the strength of Theorem~\ref{lem:UB-balancedprices} with our guiding example: XOS combinatorial auctions. Consider the pricing rule from Equation~\eqref{eq:pricing-rule-XOS} and extend it linearly to bundles. Next, \cite{Dutting2020} prove that this rule in the budget-free setting is $(1,1)$-balanced w.r.t $\OPT$ and the exchange-compatible family $(\mathcal F_{\vec x})_{\vec x\in X}$ where 
\begin{equation}\label{eq:exchange-combinatorial}\mathcal F_{\vec x}=\left\{\vec y\in \mathcal F\mid \left(\bigcup_{i\in [n]} y_i\right) \cap\left( \bigcup_{i\in [n]}x_i\right)=\emptyset \right\}.
\end{equation} 
This then clearly extends to being balanced in the sense of definition~\ref{def:balancedprices} since it is independent of $\vec \t$. Using Theorem~\ref{lem:UB-balancedprices} with $\alpha=1=\beta_1$ and $\beta_2=0$, we see that the pricing rule $\delta \vec p$ with $\delta<1$ is $(1-\delta,\frac{1-\delta}{\delta})$-UB, like we had already seen in Section~\ref{sec:combinatorial-auctions}.

\section{Conclusion}\label{sec:conclusion}
We initiate the study of online auctions in the prophet inequality model when agents also have stochastic objectives and budgets. In the setting of XOS combinatorial auctions we derive tight bounds on the competitive ratio of a natural class of pricing rules and show, in the budget-free setting, that the competitive ratio of the half-pricing rule interpolates linearly between 2 and 4 in the price sensitivity of the least price sensitive agent.
On the way, we introduce the concept of utility-bounded prices which capture many pricing rules that appear in the literature, e.g. any balanced pricing and allocation rule that are preserved under the budget-reduction. This allows us to establish strong connections between value maximizers and budget-constrained agents and, interestingly, show that the half-pricing rule is not optimal under intermediate objectives. This demonstrates that the additional uncertainty faced by the algorithm, despite the model's generality, causes only small deterioration in competitive ratios, which can be further bounded by the minimum price sensitivity in absence of budgets. For future study, it would be interesting to investigate whether a 2-competitive pricing rule exists for objectives beyond utility maximizers.

\section*{Acknowledgements and AI Disclosure} 

The authors used a large language model (ChatGPT 5.6) to assist with editing and improving the exposition. 
All mathematical results, proofs, and technical content were developed by the authors.

\bibliographystyle{alpha}

\newcommand{\etalchar}[1]{$^{#1}$}

\appendix

\section{Further Applications}
\subsection{Combinatorial Auctions with MPH-\texorpdfstring{$k$}{} valuations}
In this section, we cover prophet inequalities for combinatorial auctions with Maximum over Positive Hypergraphs (MPH) valuations, introduced in \cite{FeigeFIILS15}.

Given a set of items $M$, a hypergraph representation of a function $v:2^M\to \mathbb R$ is a normalized set function $w:2^M\to \mathbb R$ such that $v(S)=\sum_{T\subseteq S}w(T)$. The hypergraph representation is called positive, or non-negative, when $w(T)\geq 0$ for all $T\subseteq M$. A \emph{hyperedge} of $w$ is a set $S\subseteq M$ such that $w(S)\neq 0$ and the rank of a hypergraph representation is the largest cardinality of any hyperedge. Similarly, the rank of $v$ is the rank of the hypergraph representation of $v$ and if the rank of $w$ is at most $k$ we call $v$ a hypergraph-$k$ valuation. If the hypergraph representation is non-negative, then we refer to $v$ as a positive hypergraph-$k$ function (PH-$k$).

\begin{definition}
    A monotone set function $v:2^M\to \mathbb R_{\geq0}$ is Maximum over Positive Hypergraph-$k$ (MPH-$k$) if it can be expressed as a maximum over a set of PH-$k$ valuations. That is, there exist PH-$k$ functions $\{v_\ell\}_{\ell\in \mathcal L}$ such that for every set $S\subseteq M$,
    $$v(S)=\max_{\ell \in \mathcal L}v_\ell(S),$$
    where $\mathcal L$ is an arbitrary index set.
\end{definition}

\begin{proposition}\label{prop:MPH-k}
    Given an MPH-$k$ valuation $v$ and a constant $B\in (0,\infty)$, the budget-capped valuation $\hat v(S)=\min\{v(S),B\}$ is also MPH-$k$.
\end{proposition}
\begin{proof} 
    Let $\mathcal L$ be the index set such that $v(S)=\max_{\ell \in \mathcal L}v_{\ell}(S)$ for any $S\subseteq M$. Now, fix $S\subseteq M$ arbitrarily and define $\ell\in \mathcal L$ such that $v(S)=v_\ell(S)$. If $v(S)\neq 0$, define $\hat v_{S}(X)=\min\{1,\frac{B}{v(S)}\}v_\ell(S\cap X)$ and note that $\min\{1,\frac{B}{v(S)}\}$ is a constant as we fixed $S$. Otherwise, let $\hat v_S(X)=0$. Furthermore, since $v_\ell$ is PH-$k$, we let $w_{\ell}$ denote its hypergraph representation of rank $k$. We now show that $v_\ell(S\cap X)$ is a PH-$k$ as a function of $X$ when $v(S)\neq 0$, since when $v(S)=0$ this is trivial. Let $\hat w_{\ell,S}(T):=w_{\ell}(T)$ when $T\subseteq S$ and otherwise 0. Then, we obtain
    $$v_{\ell}(S\cap X)=\sum_{T\subseteq S\cap X}w_{\ell}(T)=\sum_{T\subseteq X}\mathbbm{1}_{T\subseteq S}w_{\ell}(T)=\sum_{T\subseteq X}\hat w_{\ell,S}(T),$$
    where, clearly, $\hat w_{\ell,S}$ is non-negative and the rank of $\hat w_{\ell,S}$ is less than or equal to that of $w_{\ell}$ which is $k$. Thus, $\hat w_{\ell,S}$ defines a hypergraph representation of $v_{\ell}(S\cap X)$ proving that $v_{\ell}(S\cap X)$ remains PH-$k$. Furthermore, we have $\hat v_S(S)=\min\{v_\ell(S),B\}=\min\{v(S),B\}=\hat v(S)$ and using $\min\{1,\frac{B}{v(S)}\}\leq 1$, we obtain for $X\subseteq M$ that $\hat v_S(X)\leq v_\ell(S\cap X)\leq v(S\cap X)\leq v(X)$ since $v$ is monotone. Similarly, by the monotonicity of $v$, we also see that $\hat v_S(X)\leq \frac{B}{v(S)}v_\ell(S\cap X)\leq B$, showing that $\hat v_S(X)\leq \hat v(X)$ for any $X\subseteq M$. Now, we define the index set $\hat{\mathcal L}=2^M$ and see for any $S\subseteq M$ that $\hat v(S)=\max_{X\in 2^M}\hat v_X(S)$, as desired.
\end{proof}

For the pricing rule, we consider an arbitrary allocation rule $\ALG$. Given a valuation and budget profile $(\vec v,\vec B)$, we write $\vec{ \tilde v}$ for the supporting hypergraph-$k$ valuations of $\vec{\hat v}$, whose existence is covered by Proposition~\ref{prop:MPH-k}, for the allocation $\ALG(\vec v,\vec B)$ and write $\hat w_i$ for the hypergraph representation of $\tilde v_i$. Now, for each item $j$, we define
\begin{equation}\label{eq:price-MPH}p^{(\vec v,\vec B)}(\{j\}):=\sum_{i\in [n]}\sum_{\substack{T\ni j\\T\subseteq \ALG_i(\vec v,\vec B)}}\hat w_i(T),\end{equation}
and extend these prices linearly to bundles. In other words, we obtain the anonymous static price $p_i^{(\vec v,\vec B)}(x)=\sum_{j\in x}p^{(\vec v,\vec B)}(\{j\})$ for all $i\in [n]$ if all items in $x$ are still available. For a partial allocation $\vec y$, we let $p_i^{(\vec v,\vec B)}(x_i\mid \vec y)=p_i^{(\vec v,\vec B)}(x_i)$ if $x_i$ is disjoint from $\vec y$, otherwise it is $\infty$.

\begin{theorem}[Theorem C.2. \cite{Dutting2020}]\label{thm:balancedness-MPH-k}
When $\vec B=\vec \infty$, the pricing rule $\vec p^{\vec v}$ defined in Equation~\eqref{eq:price-MPH} is, for any MPH-$k$ valuation profile $\vec v$, $(1,1,k-1)$-balanced w.r.t. any allocation rule $\ALG$ and exchange-compatible family of sets $(\mathcal F_{\vec x})_{\vec x\in X}$ defined in Equation~\eqref{eq:exchange-combinatorial}.
\end{theorem}

Then since $\OPT$ is preserved under the budget reduction, we see that the prices will be preserved under the reduction when choosing $\ALG=\OPT$ as its dependence on $(\vec v,\vec B)$ is equivalent to its dependence on $\vec{\hat v}$. Using Theorem~\ref{thm:balancedness-MPH-k} in combination with Theorem~\ref{lem:UB-balancedprices} and Theorem~\ref{thm:utility-bounded-guarantee} gives the following result.

\begin{theorem}\label{thm:budget-bound-MPH}
    Let $0<\delta< \frac{1}{k}$ and consider the class $\mathcal I_{\Sigma}$ of combinatorial auctions where agents have MPH-$k$ valuations, are constrained by budgets, and have price sensitivities contained in $\Sigma$. Then, for the competitive ratio of the posted price mechanism with pricing rule $\delta \vec p$ where $p_i(x_i):=\E_{(\vec v,\vec B,\vec \t)\sim \mathcal D}[p_i^{(\vec v,\vec B)}(x_i)]$ over the class $\mathcal I_{\Sigma}$, the following holds
    $$\comp(\mathcal I_{\Sigma},\delta \vec p)\leq \frac{1}{(1-\delta k)\delta}.$$
    Thus, for the (optimal) choice of $\delta=\frac{1}{2k}$ the competitive ratio is bounded above by $4k$.
\end{theorem}

When $k>1$ and agents are not budget-constrained and are utility maximizers, the competitive ratio improves to
$$\comp(\mathcal I_{\{1\}}^{\infty},\delta \vec p)\leq \frac{1-\delta}{(1-\delta k)\delta}.$$
For the choice of $\delta = \frac{1}{2k-1}$, we recover the result of $4k-2$ from \cite{Dutting2020}. Note that this is not necessarily the optimal choice, which is given $\delta = 1-\sqrt{1-\frac{1}{k}}$, yielding a slight improvement from $4k-2$ to $2k+2\sqrt{k(k-1)}-1$. This improvement was also found in \cite{braun2023simplified} using an LP-duality while here we obtain it explicitly. Interestingly, budget constraints barely worsen the upper bound for large $k$ when comparing the budget-free competitive ratio in terms of social welfare to the competitive ratio in terms of liquid welfare.

\subsection{Matroid Constrained Auctions}
In this section, we apply the results from the main text to matroids. A matroid $\mathcal M=(E,\mathcal I)$ consists of a ground set $E$ and downward-closed set $\mathcal I$ satisfying the matroid exchange axiom: for all $I,J\in \mathcal I$ such that $|I|<|J|$, there exists an element $x\in J\setminus I$ such that $I\cup \{x\}\in \mathcal I$. 
The sets in $\mathcal I$ are called independent sets. Now, for our setting, let $E$ be partitioned into subsets $E_1,\dots,E_n$, corresponding to agents, and for each $i$ we let the set of outcomes $X_i$ be $2^{E_i}$. 
An allocation $\vec x$ is feasible whenever the chosen elements form an independent set, i.e. $\bigcup_{i\in [n]}x_i\in \mathcal I$. We let $\OPT(\vec v,\vec B\mid S)$ denote the liquid welfare maximizing allocation $T\in \mathcal I$ such that $T\cap S=\emptyset$ and $T\cup S\in \mathcal I$. For this setting, we restrict ourselves to XOS valuations. 
Given a valuation and budget profile $(\vec v,\vec B)$, we write $\vec{\tilde v}$ for the supporting additive valuations of $\vec{\hat v}$ for the allocation $\OPT(\vec v,\vec B)$. Then, letting $Y=\bigcup_{j\in [n]}y_j$, the pricing rule is given
\begin{equation}\label{eq:price-matroid}p_i^{(\vec v,\vec B)}(x_i\mid \vec y)=\begin{cases}\sum_{\ell\in [n]}\tilde v_\ell\left(\OPT_\ell\left(\vec{\tilde v}\mid Y\right)\right)-\sum_{\ell\in [n]}\tilde v_\ell\left(\OPT_\ell\left(\vec{\tilde v}\mid Y\cup x_i\right)\right), &\text{if $Y\cup x_i\in \mathcal I$}\\
\infty, &\text{otherwise}.
\end{cases}\end{equation}

For additive valuations in budget-free settings, this pricing rule reduces to the one in presented in $\cite{Dutting2020}$ and is proven over budget-free instances to be $(1,1)$-balanced w.r.t. $\OPT$. They extended these results to XOS valuations giving the following theorem.

\begin{theorem}[Theorem E.1. \cite{Dutting2020}]\label{thm:balanced-matroids}
    Let $\vec v$ be any XOS valuation profile and $\vec B=\vec \infty$. Then the pricing rule $\vec p^{\vec v}$ defined above is for any XOS-valuation profile $\vec v$, $(1,1)$-balanced w.r.t. $\OPT$ and exchange-compatible family of sets $(\mathcal F_{\vec x})_{\vec x\in X}$ defined by
    \begin{equation}\label{eq:exchangecompatible-matroid}\mathcal F_{\vec x}:=\left\{\vec y\in \prod_{i\in [n]}2^{E_i}\mid\left(\bigcup_{i\in [n]} y_i\right) \cap\left( \bigcup_{i\in [n]}x_i\right)=\emptyset \; \text{and}\; \left(\bigcup_{i\in [n]} y_i\right) \cup\left( \bigcup_{i\in [n]}x_i\right)\in \mathcal I\right\}.\end{equation}
\end{theorem}

Next, since this pricing rule only depends on $(\vec v,\vec B)$ through $\vec{\hat v}$, we see that the pricing rule is preserved under the reduction. Furthermore, for $\mathcal F_{\vec x}$ defined in Equation~\eqref{eq:exchangecompatible-matroid}, we see that $\mathcal F_{\vec x}\subseteq \mathcal F$ as $\vec y\in \mathcal F_{\vec x}$ implies by the downward-closed property of $\mathcal I$ that $\bigcup_{i\in [n]}y_i\in \mathcal I$. Thus, the conditions of Theorem~\ref{lem:UB-balancedprices} are satisfied and applying it in combination with Theorems~\ref{thm:balanced-matroids}~\&~\ref{thm:utility-bounded-guarantee} gives the following result.
\begin{theorem}\label{thm:budget-bound-matroid}
    Let $\mathcal I_{\Sigma}$ be a class of instances where agents have price sensitivities contained in $\Sigma$, are constrained by budgets and have XOS valuations. Then, for $0<\delta<1$, the competitive ratio of the posted price mechanism with pricing rule $\delta \vec p$ where $p_i(x_i|\vec y):=\E_{(\vec v,\vec B,\vec \t)\sim \mathcal D}[p_i^{(\vec v,\vec B)}(x_i|\vec y)]$ over the class $\mathcal I_{\Sigma}$ satisfies the following bound
    $$\comp(\mathcal I_{\Sigma},\delta \vec p)\leq \frac{1}{(1-\delta)\delta}.$$
    For $\delta=\frac12$ this ratio is minimized and gives $4$.
\end{theorem}

As noted in \cite{Dutting2020}, the construction is NP-hard for submodular or XOS valuations. Let $\GRD$ denote the allocation algorithm that allocates items greedily by allocating items that locally increase $\LW(\vec x)$ the most subject to the matroid constraint. Note that this is equivalent to greedily allocating items that locally increase $\SW(\vec x)$ w.r.t. the valuation profile $\vec{\hat v}$. Thus, since $\hat v(\cdot)=\min\{v(\cdot),B\}$ is submodular whenever $v(\cdot)$ is, we see that $\GRD$ is preserved under the budget reduction and that it remains a $2$-approximation to $\OPT$ for submodular valuations in budget-constrained settings.
In \cite{Dutting2020}, using their composability results, they prove for the pricing rule in Equation~\eqref{eq:price-matroid} over budget-free realizations that if the additive support is taken for the allocation algorithm $\GRD(\vec v,\vec B)$ instead of $\OPT(\vec v,\vec B)$, the pricing rule becomes $(1,1)$-balanced w.r.t. $\GRD$ for any submodular valuation profile $\vec v$. Combining this with Theorem~\ref{lem:UB-balancedprices} and Theorem~\ref{thm:utility-bounded-guarantee} immediately gives computationally feasible balanced prices under the oracle assumptions as \cite{Dutting2020}. The resulting posted price mechanism yields an $8$-approximation to the expected optimal liquid welfare when agents are constrained by budgets and have submodular valuations, up to an additive sampling error.

\section{Missing Proofs}\label{sec:missing-proofs}
\subsection*{Proof of Lemma~\ref{lemma:objectivebound}}
\lemmaobjective*
\begin{proof}
    We write $\vec z'(\vec v,\vec B,\vec \t,\vec v',\vec B'):=\OPT(\vec v',\vec B',\mathcal F_{\vec z(\vec v,\vec B,\vec \t)})$ for the optimal allocation with the valuation and budget profiles $(\vec v',\vec B')$ under the feasibility constraint $\mathcal F_{\vec z(\vec v,\vec B,\vec \t)}$. We sample $(\vec v',\vec B',\vec \t')\sim \mathcal D$ and define
    $$\tilde y_i^{\vtype}(\vtype'):= \OPT_i\left((v_i,\vec v_{-i}'),(B_i,\vec B_{-i}'),\mathcal F_{\vec z\left((v_i',\vec v_{-i}),(B_i',\vec B_{-i}),(\t_i',\vec \t_{-i})\right)}\right),$$
    of which the observed price is
    $$\delta\cdot p_i\left(\OPT_i\left((v_i,\vec v_{-i}'),(B_i,\vec B_{-i}'),\mathcal F_{\vec z\left((v_i',\vec v_{-i}),(B_i',\vec B_{-i}),(\t_i',\vec \t_{-i})\right)}\right)\mid \vec z_{[i-1]}(\vec v,\vec B,\vec \t)\right).$$ 
    So, let $Y_i^{\vtype}$ be the distribution of $y_i^{\vtype}(\vtype')$ where $\vtype'\sim \mathcal D$. Suppose that agent $i$ deviates unilaterally from choosing $\vec z(\vec v,\vec B,\vec \t)$ to buying
    $$y_i^{\vtype}(\vtype'):= \begin{cases}
        \tilde y_i^{\vtype}(\vtype'),& \text{if $\delta p_i( \tilde y_i^{\vtype}(\vtype') \; \mid \; \vec z_{[i-1]}(\vtype))\leq \hat v_i( \tilde y_i^{\vtype}(\vtype'))$}\\
        \emptyset & \text{else}.
    \end{cases}$$
    Notice that this deviation satisfies the RoS and budget constraints. Furthermore, the bundle is available since
    $$\OPT\left((v_i,\vec v_{-i}'),(B_i,\vec B_{-i}'),\mathcal F_{\vec z\left((v_i',\vec v_{-i}),(B_i',\vec B_{-i}),(\t_i',\vec \t_{-i})\right)}\right)\in \mathcal F_{\vec z\left((v_i',\vec v_{-i}),(B_i',\vec B_{-i}),(\t_i',\vec \t_{-i})\right)}$$ 
    and because of the definition of exchange compatible. Thus, the deviation is feasible.
    In that case, we see that 
    $$\begin{aligned}\E_{\vtype,\vtype'\sim \mathcal D}[\hat u_i(y_i^{\vtype}(\vtype'))]&\geq
    \E_{\substack{(\vec v,\vec B,\vec \t)\sim \mathcal D\\ (\vec v',\vec B',\vec \sigma')\sim \mathcal D}}\bigg[\hat v_i\left(\OPT_i\left((v_i,\vec v_{-i}'),(B_i,\vec B_{-i}'),\mathcal F_{\vec z\left((v_i',\vec v_{-i}),(B_i',\vec B_{-i}),(\t_i',\vec \t_{-i})\right)}\right)\right)\\ 
    &-\delta\cdot p_i\left(\OPT_i\left((v_i,\vec v_{-i}'),(B_i,\vec B_{-i}'),\mathcal F_{\vec z\left((v_i',\vec v_{-i}),(B_i',\vec B_{-i}),(\t_i',\vec \t_{-i})\right)}\right)\mid \vec z_{[i-1]}(\vec v,\vec B,\vec \t)\right)\bigg]\\
    &= \E_{\substack{(\vec v,\vec B,\vec \t)\sim \mathcal D\\ (\vec v',\vec B',\vec \sigma')\sim \mathcal D}}\bigg[\hat v_i'(\OPT_i(\vec v',\vec B',\mathcal F_{\vec z(\vec v,\vec B,\vec \t)}))\\
    &-\delta\cdot p_i\left(\OPT_i(\vec v',\vec B',\mathcal F_{\vec z(\vec v,\vec B,\vec \t)})\mid\vec z_{[i-1]}(\vec v,\vec B,\vec \t)\right)\bigg]\\
    &=\E_{\substack{(\vec v,\vec B,\vec \t)\sim \mathcal D\\ (\vec v',\vec B')\sim \pi}}\bigg[\hat v_i'\left( z_i'(\vec v,\vec B,\vec \t,\vec v',\vec B')\right)-\delta\cdot p_i\left( z_i'(\vec v,\vec B,\vec \t,\vec v',\vec B')\mid\vec z_{[i-1]}(\vec v,\vec B,\vec \t)\right)\bigg],\end{aligned}$$
    where in the second equality we used that $\vec z_{[i-1]}(\vec v,\vec B,\vec \t)$ does not depend on $v_i$ and $B_i$ and $\sigma_i$. Furthermore, in the last equality, we drop the dependence on $\vec \sigma'$.
    Summing over all agents gives    \begin{equation}\label{eq:continue1}\begin{aligned}\sum_{i\in [n]}\E_{\vtype,\vtype'\sim \mathcal D}[\hat u_i(y_i^{\vtype}(\vtype'))]&\geq \E_{\substack{(\vec v,\vec B,\vec \t)\sim \mathcal D\\ (\vec v',\vec B')\sim \pi}}\left[\SW_{\vec{\hat v}'}(\OPT(\vec v',\vec B',\mathcal F_{\vec z(\vec v,\vec B,\vec \t)}))\right]\\
    &\qquad -\E_{\substack{(\vec v,\vec B,\vec \t)\sim \mathcal D\\ (\vec v',\vec B')\sim \pi}}\left[\sum_{i\in[n]}\delta\cdot p_i\left( z_i'(\vec v,\vec B,\vec \t,\vec v',\vec B')\mid\vec z_{[i-1]}(\vec v,\vec B,\vec \t)\right)\right]
    \end{aligned}\end{equation}

    Now, note that  $\vec z'(\vec v,\vec B,\vec \t,\vec v',\vec B')\in \mathcal F_{\vec z(\vec v,\vec B,\vec\t)}$. Thus, by the second inequality of balanced prices in Equation~\eqref{eq:low-enough}, we see that
    $$\begin{aligned}\sum_{i\in[n]}\delta p_i\left( z_i'(\vec v,\vec B,\vec \t,\vec v',\vec B')\mid \vec z_{[i-1]}(\vec v,\vec B,\vec \t)\right)&=\sum_{i\in [n]}\delta \E_{(\tilde{\vec v},\tilde{\vec \t})\sim \hat{\mathcal D}}\left[p^{(\tilde{\vec v},\tilde{\vec \t})}\left( z_i'(\vec v,\vec B,\vec \t,\vec v',\vec B')\mid \vec z_{[i-1]}(\vec v,\vec B,\vec \t)\right)\right]\\
    &\leq \delta \beta_1\E_{(\tilde{\vec v},\tilde{\vec \sigma})\sim \hat{\mathcal D}}\left[\SW_{\tilde{\vec v}}(\OPT(\tilde{\vec v},\mathcal F_{\vec z(\vec v,\vec B,\vec \t)}))\right]\\
    &\qquad+\delta\beta_2\E_{(\tilde{\vec v},\tilde{\vec \t})\sim \hat{\mathcal D}}\left[\SW_{\tilde{\vec v}}(\ALG(\tilde{\vec v},\tilde{\vec \t}))\right],
    \end{aligned}$$
    where first equality holds because the pricing rule is preserved under the reduction.

    Using this and continuing from the inequality in Equation~\eqref{eq:continue1} gives
    $$\begin{aligned}\sum_{i\in [n]}\E_{\vtype,\vtype'\sim \mathcal D}[\hat u_i(y_i^{\vtype}(\vtype'))]&\geq \E_{\substack{(\vec v,\vec B,\vec \t)\sim \mathcal D\\ (\vec v',\vec B')\sim \pi}}\left[\SW_{\vec{\hat v}'}(\OPT(\vec v',\vec B',\mathcal F_{\vec z(\vec v,\vec B,\vec \t)}))\right]\\
    &\qquad -\delta \beta_1\E_{\substack{(\vec v,\vec B,\vec \t)\sim \mathcal D\\ (\tilde{\vec v},\tilde{\vec \t})\sim \hat{\mathcal D}}}\left[\SW_{\tilde{\vec v}}(\OPT(\tilde{\vec v},\mathcal F_{\vec z(\vec v,\vec B,\vec \t)}))\right]\\
    &\qquad -\delta \beta_2\E_{ (\tilde{\vec v},\tilde{\vec \t})\sim\hat {\mathcal D}}\left[\SW_{\tilde{\vec v}}(\ALG(\tilde{\vec v},\tilde{\vec \t}))\right].
    \end{aligned}$$
    Now, note that since the realizations of $\hat{\mathcal D}$ are coupled to the realizations of $\mathcal D$, we see that
    $$\E_{ (\tilde{\vec v},\tilde{\vec \t})\sim\hat {\mathcal D}}\left[\SW_{\tilde{\vec v}}(\ALG(\tilde{\vec v},\tilde{\vec \t}))\right]=\E_{(\vec v,\vec B,\vec \sigma)\sim \mathcal D}\left[\SW_{\vec{\hat v}}(\ALG(\vec{\hat v},\vec{\hat \sigma}))\right]=\E_{(\vec v,\vec B,\vec \sigma)\sim \mathcal D}\left[\LW(\ALG(\vec v,\vec B,\vec \sigma))\right],$$
    where in the last equality we used that $\ALG$ is preserved under the reduction.
    Furthermore, by an analogous line of reasoning, we obtain that
    $$\begin{aligned}\E_{\substack{(\vec v,\vec B,\vec \t)\sim \mathcal D\\ (\tilde{\vec v},\tilde{\vec \t})\sim \hat{\mathcal D}}}\left[\SW_{\tilde{\vec v}}(\OPT(\tilde{\vec v},\mathcal F_{\vec z(\vec v,\vec B,\vec \t)}))\right]&=\E_{\substack{(\vec v,\vec B,\vec \t)\sim \mathcal D\\ (\vec v',\vec B',\vec \t')\sim \mathcal D}}\left[\SW_{\vec{\hat v}'}(\OPT(\vec{\hat v}',\mathcal F_{\vec z(\vec v,\vec B,\vec \t)}))\right]\\
    &=\E_{\substack{(\vec v,\vec B,\vec \t)\sim \mathcal D\\ (\vec v',\vec B')\sim \pi}}\left[\LW_{(\vec v',\vec B')}(\OPT(\vec v',\vec B',\mathcal F_{\vec z(\vec v,\vec B,\vec \t)}))\right].
    \end{aligned}$$
    Similarly, the first term on the RHS of the inequality in Equation~\eqref{eq:continue1} can be turned into liquid welfare and will equal to the above. Combining everything together, one obtains
    $$\begin{aligned}\sum_{i\in [n]}\E_{\vtype,\vtype'\sim \mathcal D}[\hat u_i(y_i^{\vtype}(\vtype'))]&\geq (1-\delta \beta_1)\E_{\substack{\vtype\sim \mathcal D\\ (\vec v',\vec B')\sim \pi}}\left[\LW_{(\vec v',\vec B')}(\OPT(\vec v',\vec B',\mathcal F_{\vec z(\vtype)}))\right]-\delta \beta_2\E_{\vtype\sim \mathcal D}\left[\LW(\ALG(\vtype))\right].
    \end{aligned}$$
\end{proof}

\subsection*{Proof of Lemma~\ref{lemma:revenuebound}}
\lemmarevenue*
\begin{proof}[Proof of Lemma~\ref{lemma:revenuebound}]
Immediately applying the first property of balanced prices in Equation~\eqref{eq:high-enough} gives
$$\begin{aligned}\sum_{i\in[n]}\delta p_i\left( z_i(\vec v,\vec B,\vec \t)\mid \vec z_{[i-1]}(\vec v,\vec B,\vec \t)\right)&=\sum_{i\in [n]}\delta \E_{(\tilde{\vec v},\tilde{\vec \t})\sim \hat{\mathcal D}}\left[p^{(\tilde{\vec v},\tilde{\vec \t})}\left( z_i(\vec v,\vec B,\vec \t)\mid \vec z_{[i-1]}(\vec v,\vec B,\vec \t)\right)\right]\\
&\geq \frac{\delta}{\alpha}\E_{(\tilde{\vec v},\tilde{\vec \sigma})\sim \hat{\mathcal D}}\left[\SW_{\tilde{\vec v}}(\ALG(\tilde{\vec v},\tilde{\vec \t}))\right]\\
&\qquad-\frac{\delta}{\alpha}\E_{(\tilde{\vec v},\tilde{\vec \t})\sim \hat{\mathcal D}}\left[\SW_{\tilde{\vec v}}(\OPT(\tilde{\vec v},\mathcal F_{\vec z(\vec v,\vec B,\vec \t)}))\right].
\end{aligned}$$
Here, the first equality holds since the pricing rule is preserved under the reduction. Next, taking the expectation on both sides and rewriting the expectations like in the proof of Lemma~\ref{lemma:objectivebound} gives the result.
\end{proof}

\end{document}